\documentclass[final,3p,times,12pt]{elsarticle}

\usepackage{amsmath}
\usepackage{graphicx}
\usepackage{amsfonts}

\usepackage{amssymb}
\usepackage{array}
\usepackage{mathrsfs}
\usepackage{multirow}
\usepackage{amsmath,amssymb,amsthm}
\usepackage{amsthm}
\usepackage{subfig}
\usepackage{color}

\newtheorem{thm}{Theorem}[section]

\newtheorem{lem}{Lemma}[section]

\newtheorem{cor}{Corollary}[section]

\newtheorem{rem}{Remark}[section]

\newtheorem{exa}{Example}[section]

\def\<{\left<}\def\>{\right>}
\def\({\left(}\def\){\right)}

\allowdisplaybreaks

\begin{document}

\begin{frontmatter}

\title{Generalized Expected Discounted Penalty Function at General Drawdown for L\'{e}vy Risk Processes\tnoteref{thanks}}
\tnotetext[thanks]{This work was partially supported by the National Natural Science Foundation of China (Nos. 11661074 and 11701436), the Program for New Century Excellent Talents in Fujian Province University (No. Z0210103) and the Fundamental Research Funds for the Central Universities (Nos. 20720170096 and 2018IB019).}

\author[]{Wenyuan WANG$^a$}
\ead{wwywang@xmu.edu.cn}

\author[]{Ping CHEN$^b$\corref{cor1}}
\ead{pche@unimelb.edu.au}\cortext[cor1]{Corresponding author. Tel:+61 90358053}

\author{Shuanming LI$^b$}
\ead{shli@unimelb.edu.au}

\address[a]{School of Mathematical Sciences, Xiamen University, Xiamen 361005, P. R. China.}
\address[b]{Department of Economics, The University of Melbourne, Parkville,
Victoria 3010, Australia.}

\begin{abstract}

This paper considers an insurance surplus process modeled by a spectrally negative L\'{e}vy process. Instead of the time of ruin in the traditional setting, we apply the time of drawdown as the risk indicator in this paper. We study the joint distribution of the time of drawdown, the running maximum at drawdown, the last minimum before drawdown, the surplus before drawdown and the surplus at drawdown (may not be deficit in this case), which generalizes the known results on the classical expected discounted penalty function in Gerber and Shiu (1998). The results have semi-explicit expressions in terms of the $q$-scale functions and the L\'{e}vy measure associated with the L\'{e}vy process.
As applications, the obtained result is applied to recover results in the literature and to obtain new results for the Gerber-Shiu function at ruin for risk processes embedded with a loss-carry-forward taxation system or a barrier dividend strategy. Moreover, numerical examples are provided to illustrate the results.
\end{abstract}


\end{frontmatter}


\section{Introduction }
\setcounter{section}{1}

In the classical model of risk theory, the behavior of the insurer's risk process is analysed through the expected discounted penalty function, which is commonly referred to as Gerber-Shiu function in the ruin literature, see Gerber and Shiu (1998). Based on the Cram${\rm \acute{e}}$r-Lundberg model for the surplus process, they studied the joint distribution of three key quantities of interest: the probability of ruin, the distribution of surplus immediately prior to ruin and deficit at time of ruin. Thereafter the Gerber-Shiu function has been studied extensively in the literature. The diversity of techniques and perspectives under which this function was studied initiated a special issue of \emph{Insurance: Mathematics and Economics} on the topic of Gerber-Shiu functions in 2010.

Motivated by the recent development in risk measurements, this paper studies an extended definition of the classical Gerber-Shiu function. The main idea is to replace the time of ruin by a more general risk indicator: the time of drawdown, which is widely used in industry. Generally, a drawdown time refers to a moment when the surplus process declines from the peak to the subsequent trough during a specific recorded period of an investment, fund or commodity security. We adopt a general drawdown definition which includes not only the ruin time as a special case, but also many other forms (linear or non-linear), see Remark 1 and Remark 2 in Avram et al. (2017) for examples and their explanations. In terms of the surplus process of an insurer, the sooner the drawdown time occurs, the more risk the insurer is bearing. Accordingly, the surplus before drawdown and the surplus at drawdown also play an important role in determining an insurer's financial risk. The level of those surpluses can help manage the financial decision of the insurer. For example, an insurer may try to minimize the probability of drawdowns of 20\% or greater before increasing its premium rate to avoid even worse situitions. The mathematical formulation was first introduced by  Taylor (1975) who studied the maximum drawdown of a drifted Brownian motion. This result was later extended to other situations, see Avram et al. (2004), Landriault et al. (2015), Landriault et al. (2017), Li et al. (2017), Wang and Zhou (2018) and the references therein.

In the context of finance and actuarial studies, the application of drawdown risks has been flourishing in recent years. Shepp and Shiryaev (1993) proposed a new put option where the option buyer receives the maximum price that the option has ever traded from the purchase time and the exercise time. More recent applications in option pricing can be found in Avram et al. (2004) and Carr (2014). In terms of portfolio selection, Grossman and Zhou (1993) pioneered this research topic by adopting a strict drawdown constraint on the optimal investment strategy. Extended work along this line is abundant, to name a few but not limited to, see Cvitanic and Karatzas (1995) for a multi-asset framework, Cherny and Obloj (2013) for a general semimartingale framework, Roche (2006) for an optimal consumption-investment problem, and Elie and Touzi (2008) for the optimization over a general class of utility functions. Along another line in the portfolio selection, the probability of drawdown is minimized instead of imposing a constraint on drawdown. Various scenarios have been considered, see Chen et al. (2015) for a pure investment formulation, Angoshtari et al. (2016) for a case with constant consumption constraint, and Han et al. (2018) for an optimal reinsurance case. In terms of dividend optimization problems, Wang and Zhou (2018) considered a general version of de Finetti's optimal dividend problem in which the ruin time is replaced with a general drawdown time.

Another feature of this paper is the insurer's surplus process is modelled by a spectrally negative L\'{e}vy process, which is a stochastic process with stationary independent increments and with sample paths of no positive jumps. It often serves as a surplus process in risk theory, where the downward jumps describe the outgoing payments of claims. The application of spectrally negative L\'{e}vy processes in risk theory can be seen in Yang and Zhang (2001), Garrido and Morales (2006), Biffis and Morales (2010), Avarm et al. (2017) and Loeffen et al. (2018). Based on the time of drawdown, this paper studies an extended definition of the expected discounted penalty function in terms of the $q$-scale functions and the L\'{e}vy measure associated with the L\'{e}vy process. The joint distribution of the time of drawdown,
the running maximum at drawdown, the last minimum before drawdown,
the surplus before drawdown and the surplus at drawdown is derived. Unlike the time of ruin, the surplus process may not fall below zero at the time of drawdown, therefore, we use the surplus at drawdown instead of the deficit at ruin in the extended definition of Gerber-Shiu function.

From technical point of view, the classical ruin theory is mainly based on renewal equation techniques to obtain some delicate results, see Gerber and Shiu (1998), or some specific methods for some particular risk models such as a Gamma process, see Dufresne and Gerber (1993). When the risk process is extended to a spectrally negative L\'{e}vy process, the results on ruin probabilities follow from the fluctuation theory of L\'{e}vy processes, see Kyprianou (2006).
In the case of ruin, different mathematical subtleties from the fluctuation theory are used for specific model formulation, such as the Laplace exponent when the classical risk process and the gamma process are perturbed by diffusion, see Yang and Zhang (2001); the Laplace transform of the time to ruin when the aggregate claims process is a subordinator, see Garrido and Morales (2006); the overshoots results when the Parisian ruin problem is investigated, see Loeffen et al. (2018). The ruin probability is usually expressed in terms of a multi-fold convolution of some distribution functions, or the $q$-scale functions associated to the L\'{e}vy process.

The results of this paper (the case of drawdown) are derived based on the excursion-theoretical approach, 
which is also from the fluctuation theory of L\'{e}vy processes and has been proving its efficiency in solving the related boundary crossing problems. Using this approach, Kyprianou and Pistorius (2003)
derived the Laplace transform of a crossing time which is the key quantity to the evaluation of the
Russian option; Avram et al. (2004) determined the joint Laplace transform of the exit time and exit position from an interval containing the origin of the process reflected in its supremum, which is then applied to an optimal stopping problem associated with the pricing of Russian options and their Canadized versions; Pistorius (2004) derived the $q$-resolvent kernels for the L\'{e}vy process reflected at its supremum killed upon leaving $[0, a]$; Pistorius (2007) solved the problem of Lehoczky and the Skorokhod embedding problem for the L\'{e}vy process reflected at its supremum;
Baurdoux (2007) investigated the density of the resolvent measure of the killed L\'{e}vy process reflected at its infimum;
Kyprianou and Zhou (2009) obtained the generalized version of the Gerber-Shiu function for a taxed L\'{e}vy risk process where a loss-carry-forward taxation system is embedded.


The existing literature has been witnessing the applications of the powerful excursion-theoretical approach in the fields of financial mathematics and stochastic process theory. However, as far as the authors know, this approach is rarely used in the field of actuarial risk theory. This paper attempts to apply the approach to study the expression of the extended Gerber-Shiu function in the case of drawdown. One merit of applying the excursion-theoretical approach is, we do not need to use specific features of the underlying L\'{e}vy process except for a generic path decomposition
in terms of excursions from the running maximum, which falls into the framework of Poisson point process.
Thence we are allowed to use the theory of Poisson point process, such as the compensation formula, in the whole manipulation of our target problem. We mention that all results in this paper are expressed elegantly in terms of scale functions and the L\'{e}vy measure associated with the L\'{e}vy process. We also mention that, when the first passage over the non-constant general drawdown boundary is reduced to the case of constant boundary, we are able to recover the corresponding results in the existing literature.

We point out that Li et al. (2017) also applied the excursion approach to study the exit problems involving a general drawdown time for spectrally negative L\'{e}vy processes. However their focus was on the joint Laplace transform for the process at the drawdown time, while this paper aims to study the joint distribution involving general drawdown times. It is true that the joint distribution of random quantities is uniquely determined by the corresponding Laplace transform, and can be obtained by inverting the Laplace transform either analytically by the Bromwich integral or numerically. Unfortunately, many problems of mathematical interest or physical interest lead to Laplace transforms whose inverses are not readily expressed in terms of tabulated functions. And also, to the best of our knowledge, all the current numerical inversion methods are unstable in the sense that small \lq\lq input\rq\rq  errors, arising say from computer roundoff or from parameter selection inherent in the algorithm, can be disastrously magnified in the inversion, which should be avoided in some practical problems such as the survival probability of a population in a finite capacity environment or in tumor growth models, see Albano and Giorno (2006) for an example. Even with the development of the modern technology, it is hard to find a universal algorithm that works well to all the cases. A nice review of these methods can be found in Davies (2002). Therefore it is still worthwhile to study the joint distribution even when the Laplace transform is readily available.

Besides, in the computation of the joint distribution we add a constraint on the surplus level which is another difference between our paper and Li et al (2017). The motivation of this constraint comes from Biffis and Morales (2010) where the last minimum of the surplus before ruin was considered to better reflect the company's financial condition. In our paper we extend this terminology to the last minimum of the surplus before drawdown, which serves as a warning line of the company's financial activities. We refer to Remark 3.1 for more detailed interpretations.

This paper is organized as follows. Section 2 presents some preliminary facts on spectrally negative L{\'e}vy processes. The main results, proofs and discussions are provided in Section 3. In Section 4, the main results are applied to study the Gerber-Shiu function at ruin for L\'evy risk processes with tax and dividends.
Section 5 provides some numerical examples to illustrate our results. Section 6 concludes this paper.

\section{Preliminaries of spectrally negative L\'{e}vy process}
\setcounter{section}{2}

Write $X=\{X(t);t\geq0\}$, defined on a probability space with probability laws $\{\mathbb{P}_{x};x\in\mathbf{R}\}$ and natural filtration $\{\mathcal{F}_{t};t\geq0\}$, for a spectrally negative L\'{e}vy process.
We denote its running supremum and running infimum process, respectively, as $\{\overline{X}(t)=\sup\limits_{0\leq s\leq t}X(s);t\geq0\}$ and $\{\underline{X}(t)=\inf\limits_{0\leq s\leq t}X(s);t\geq0\}$.

A function $\xi$ defined on $(0,\infty)$ is called a general drawdown function if it is continuous and $\overline{\xi}(x)=x-\xi(x)>0$ for $x>0$.
The general drawdown time with respect to the general drawdown function $\xi(\cdot)$, also called the $\xi$-drawdown time for short, is defined as
$$\tau_{\xi}=\inf\{t\geq0;X(t)<\xi(\overline{X}(t))\}.$$
\begin{rem}
Note that when $\xi(\cdot)\equiv 0$, $\tau_{\xi}$ reduces to the classical ruin time. Another example is a linear function of the running supreme, say, $\xi(\overline{X}(t))=0.8 \overline{X}(t)-0.5$. Then $X(t)<\xi(\overline{X}(t))$ is equivalent to $0.8 \overline{X}(t)-X(t)>0.5$. Accordingly, $\tau_{\xi}$ refers to the first time the surplus process drops $0.5$ units below 80\% of its maximum to date. In practice, the risk manager can use it as a turning point of taking some actions, such as increasing the premium rate or negotiating by the company with the capital providers to prevent even worse situations. For non-linear forms of $\xi(\cdot)$, we refer to Remark 2 in Avram et al. (2017) for the examples and their explanations.
\end{rem}
We also define the first down-crossing time of level $a$ and up-crossing time of level $b$, respectively, as follows
\begin{eqnarray}
\tau_{a}^{-}:=\inf\{t\geq0;X(t)<a\}\,\,\,\text{and}\,\,\,\tau^+_{b}:=\inf\{t\geq0;X(t)>b\}.\nonumber
\end{eqnarray}

Let the Laplace exponent of $X$ be given by
\begin{eqnarray}
\psi(\theta)=\ln \mathbb{E}_{x}\left(\mathrm{e}^{\theta (X(1)-x)}\right)=\gamma\theta+\frac{1}{2}\sigma^{2}\theta^{2}-\int_{(0,\infty)}\left(1-\mathrm{e}^{-\theta x}-\theta x\mathbf{1}_{(0,1)}(x)\right)\upsilon(\mathrm{d}x),\nonumber
\end{eqnarray}
where $\upsilon$ is  the L\'{e}vy measure satisfying $\int_{(0,\infty)}\left(1\wedge x^{2}\right)\upsilon(\mathrm{d}x)<\infty$.
It is known that $\psi(\theta)$ is finite for  $\theta\in[0,\infty)$ in which case it is strictly convex and infinitely differentiable.
As in Bertoin (1996), the $q$-scale functions $\{W_{q};q\geq0\}$ of $X$ are defined as follows. For each $q\geq0$, $W_{q}:\,[0,\infty)\rightarrow[0,\infty)$ is the unique strictly increasing and continuous function with Laplace transform
\begin{eqnarray}
\int_{0}^{\infty}\mathrm{e}^{-\theta x}W_{q}(x)\mathrm{d}x=\frac{1}{\psi(\theta)-q},\quad \mbox{for }\theta>\Phi_{q},\nonumber
\end{eqnarray}
where $\Phi_{q}$ is the largest solution of the equation $\psi(\theta)=q$. Further define $W_{q}(x)=0 $ for $x<0$, and write $W$ for short for the $0$-scale function $W_{0}$.
\begin{rem}
Scale functions appear in the vast majority of known identities concerning boundary crossing problems and related path decompositions.
This in turn has consequences for their use in a number of classical applied probability models which rely heavily on such identities. We refer to Kuznetsov et al. (2012) for intuitive examples and explanations. We have to point out that for all spectrally negative L\'{e}vy processes, $q$-scale functions exist for all $q\geq 0$. However, since the scale function is defined via its Laplace transform, and in most cases it is not possible to find an explicit expression for the inverse of a Laplace transform, hence the expression in terms of the $q$-scale functions can be called semi-explicit expression. For those who cannot find explicit expressions, we resort to numerical methods which allow one to compute scale functions easily and efficiently, see also Kuznetsov et al. (2012).
\end{rem}

We  also briefly recall concepts in excursion theory for the reflected process $\{\overline{X}(t)-X(t);t\geq0\}$, and we refer to Bertoin (1996) for more details.
The process $\{L(t):= \overline{X}(t)-x, t\geq0\}$ serves as a local time at $0$ for
the Markov process $\{\overline{X}(t)-X(t);t\geq0\}$ under $\mathbb{P}_{x}$.
Let the corresponding inverse local time be defined as
$$L^{-1}(t):=\inf\{s\geq0\mid L(s)>t\}=\sup\{s\geq0\mid L(s)\leq t\}.$$
Let further $L^{-1}(t-)=\lim\limits_{s\uparrow t}L^{-1}(s)$.
The Poisson point process of excursions indexed by this local time is  denoted by $\{(t, \varepsilon_{t}); t\geq0\}$, where
$$\varepsilon_{t}(s):=X(L^{-1}(t))-X(L^{-1}(t-)+s), \,\,s\in(0,L^{-1}(t)-L^{-1}(t-)],$$
whenever $L^{-1}(t)-L^{-1}(t-)>0$.
For the case of $L^{-1}(t)-L^{-1}(t-)=0 $, define $\varepsilon_{t}=\Upsilon$ with $\Upsilon$ being an additional isolated point.
Accordingly, we denote a generic excursion as $\varepsilon(\cdot)$
(or, $\varepsilon$ for short) belonging to the space $\mathcal{E}$ of canonical excursions.
The intensity measure of the process $\{(t, \varepsilon_{t}); t\geq0\}$ is given by $\mathrm{d}t\times \mathrm{d}n$ where \lq\lq$n$\rq\rq
 is a measure on the space of excursions.
The lifetime of a canonical excursion $\varepsilon$ is denoted by $\zeta$, and its excursion height is
denoted by $\overline{\varepsilon}=\sup\limits_{t\in[0,\zeta]}\varepsilon(t)$. The first passage time of a canonical excursion $\varepsilon$ will be defined
by
\begin{eqnarray}\label{exc.upcro.b}
\rho_{b}^{+}:=\inf\{t\in[0,\zeta];\varepsilon(t)>b\},
\end{eqnarray}
with the convention $\inf\emptyset=\zeta$. In addition, define
\begin{eqnarray}
\label{exc.height.b}
\overline{\alpha}_{b}:=\sup_{t\in[0,\rho_{b}^{+})}\varepsilon(t),
\end{eqnarray}
which is the excursion height prior to $\rho_{b}^{+}$.

\begin{rem}
An excursion is a segment of the path that has zero value only at its two endpoints. It refers to a maximal open time interval such that the path is away from 0. Naturally, a L\'{e}vy process can be decomposed into a sequence of excursions. Due to the stochastic nature, say, a Brownian motion (the only continuous L\'{e}vy process) started from zero can hit zero infinitely often in any time interval, then the excursions are labelled by local times rather than by the starting time of a particular excursion. Considering there are only countably many excursions, hence there are only countably many local times which pertain to an excursion. This motivates the idea of taking the set of excursions as a Poisson Point Process on local times. Therefore, a measure can be defined to describe the intensity of the Poisson point process of excursions, which is the intuitive meaning of measure \lq\lq$n$\rq\rq in our setting. For more detailed explanations, we refer to Kyprianou (2006).
\end{rem}

\section{Main results}
\setcounter{section}{3}

This section introduces the extended expected discounted penalty function at the general drawdown time for a L\'{e}vy risk process, then express it in terms of the $q$-scale functions and the L\'{e}vy measure associated with the L\'{e}vy process.

The following three technical lemmas turn out to be helpful when we derive the main results.
Lemma \ref{lemma2} characterizes the atom at $0$ of the discounted distribution law of the
overshoot at first up-crossing time of a canonical excursion $\varepsilon$ with respect to the excursion measure.
We present a proof here for self-completeness. 

\medskip
\begin{lem}\label{lemma2}
For $q>0$ and $s>x\geq 0$, we have
\begin{eqnarray}\label{impo.iden.for.n.2}
\hspace{-0.3cm}&&\hspace{-0.3cm}n\left(\mathrm{e}^{-q \rho^{+}_{s}}; \,\overline{\varepsilon}>s,\,\varepsilon(\rho^{+}_{s})=s\right)
=
\frac{\sigma^{2}}{2}\left(\frac{\left(W_{q}^{\prime}(s)\right)^{2}}{W_{q}(s)}-W_{q}^{\prime\prime}(s)\right).
\end{eqnarray}
\end{lem}

\begin{proof}[Proof\emph{:}]\,\,\,
From (14) of Pistorius (2007), we read off that
\begin{eqnarray}\label{renewal.iden.}
\mathbb{E}_{x}\left(\mathrm{e}^{-q \tau^{-}_{0}}; \,X(\tau^{-}_{0})=0\right)
\hspace{-0.3cm}&=&\hspace{-0.3cm}
\mathbb{E}\left(\mathrm{e}^{-q \tau^{-}_{-x}}; \,X(\tau^{-}_{-x})=-x\right)
\nonumber\\
\hspace{-0.3cm}&=&\hspace{-0.3cm}
\frac{\sigma^{2}}{2}\left(W_{q}^{\prime}(x)-\Phi_{q}W_{q}(x)\right).
\end{eqnarray}
By the definition of ruin, one knows that there must exist an excursion with a positive lifetime, such that $\tau_{0}^{-}$ lies in between the left and right end points of this excursion. Denote this excursion by $\varepsilon_{\theta}$ with $\theta\geq0$, then $\varepsilon_{\theta}$ is the last excursion whose left-end point is less than $\tau_{0}^{-}$, and we have $\overline{\varepsilon}_{\theta}>x+\theta$ and $\overline{\varepsilon}_{t}\leq x+t$ for all excursions $\varepsilon_{t}$ with $t<\theta$. That is, ruin does not occur during the lifetime of $\varepsilon_{t}$ with $t<\theta$, while it does occur during the lifetime of $\varepsilon_{\theta}$.
Furthermore, one can translate the ruin time and the surplus at ruin time by the excursion $\varepsilon_{\theta}$, respectively, through
$$\tau_{0}^{-}=L^{-1}(\theta-)+\rho_{x+\theta}^{+}(\theta), \quad X(\tau_{0}^{-})=x+\theta-\varepsilon_{\theta}\left(\rho_{x+\theta}^{+}(\theta)\right),$$
where $\rho_{x+\theta}^{+}(\theta)$ is defined via (\ref{exc.upcro.b}) with $\varepsilon$ replaced by $\varepsilon_{\theta}$ and $b$ replaced by $x+\theta$. Therefore, the left hand side of (\ref{renewal.iden.}) can be translated as
\begin{eqnarray}\label{renewal.iden.via.n}
\hspace{-0.3cm}&&\hspace{-0.3cm}
\mathbb{E}_{x}\left(\mathrm{e}^{-q \tau^{-}_{0}}; \,X(\tau^{-}_{0})=0\right)
\nonumber\\
\hspace{-0.3cm}&=&\hspace{-0.3cm}
\mathbb{E}_{x}\left(\sum_{\theta\geq0}\mathrm{e}^{-qL^{-1}(\theta-)}\mathbf{1}_{\{L^{-1}(\theta-)<\tau_{0}^{-}\}}
\mathrm{e}^{-q \left(\tau_{0}^{-}-L^{-1}(\theta-)\right)}\mathbf{1}_{\{S\left(L^{-1}(\theta-)\right)=\overline{X}(\tau_{0}^{-})=x+\theta\}} \mathbf{1}_{\{X(\tau_{0}^{-})=0\}}\right)
\nonumber\\
\hspace{-0.3cm}&=&
\hspace{-0.3cm}
\mathbb{E}_{x}\left(\sum_{\theta\geq0}\mathrm{e}^{-qL^{-1}(\theta-)}
\prod_{t<\theta}\mathbf{1}_{\{\overline{\varepsilon}_{t}<x+t\}}
\mathrm{e}^{-q \rho_{x+\theta}^{+}(\theta)}
\mathbf{1}_{\{\overline{\varepsilon}_{\theta}>x+\theta\}}
 \mathbf{1}_{\{\varepsilon_{\theta}\left(\rho_{x+\theta}^{+}(\theta)\right)=x+\theta\}}\right),
\end{eqnarray}
where $\varepsilon_{\theta}$ represents the last excursion prior to $\tau_{0}^{-}$, and $\rho_{x+\theta}^{+}(\theta)$ is given by (\ref{exc.upcro.b}) with $\varepsilon$ replaced by $\varepsilon_{\theta}$ and $b$ replaced by $x+\theta$.
By the compensation formula (cf., Bertoin (1996)) in the excursion theory together with (\ref{renewal.iden.via.n}), we obtain
\begin{eqnarray}\label{alter.renewal.iden.}
\hspace{-0.3cm}&&\hspace{-0.3cm}
E_{x}\left(\mathrm{e}^{-q \tau^{-}_{0}}; \,X(\tau^{-}_{0})=0\right)
\nonumber\\
\hspace{-0.3cm}&=&
\hspace{-0.3cm}
\int_{0}^{\infty}\mathbb{E}_{x}\left(\mathrm{e}^{-qL^{-1}(\theta-)}
\prod_{t<\theta}\mathbf{1}_{\{\overline{\varepsilon}_{t}<x+t\}}\right)
n\left(\mathrm{e}^{-q \rho_{x+\theta}^{+}(\theta)}
\mathbf{1}_{\{\overline{\varepsilon}_{\theta}>x+\theta,\,\varepsilon_{\theta}\left(\rho_{x+\theta}^{+}(\theta)\right)=x+\theta\}}\right)
\mathrm{d}\theta
\nonumber\\
\hspace{-0.3cm}&=&
\hspace{-0.3cm}
\int_{0}^{\infty}\mathbb{E}_{x}\left(\mathrm{e}^{-q\tau_{x+\theta}^{+}}
\mathbf{1}_{\{\tau_{x+\theta}^{+}<\tau_{0}^{-}\}}\right)
n\left(\mathrm{e}^{-q \rho_{x+\theta}^{+}}
\mathbf{1}_{\{\overline{\varepsilon}>x+\theta,\,\varepsilon\left(\rho_{x+\theta}^{+}\right)=x+\theta\}}\right)
\mathrm{d}\theta
\nonumber\\
\hspace{-0.3cm}&=&
\hspace{-0.3cm}
\int_{x}^{\infty}\frac{W^{(q)}(x)}{W^{(q)}(s)}\,
n\left(\mathrm{e}^{-q \rho_{s}^{+}}
\mathbf{1}_{\{\overline{\varepsilon}>s,\,\varepsilon\left(\rho_{s}^{+}\right)=s\}}\right)
\mathrm{d}s
.
\end{eqnarray}
Equating (\ref{renewal.iden.}) and (\ref{alter.renewal.iden.}) and then differentiating both side of the resulting equation with respect to $x$ yields
(\ref{impo.iden.for.n.2}).
\end{proof}

\medskip
We recall from (\ref{exc.upcro.b}) and (\ref{exc.height.b}) that, $\rho_{x}^{+}$ and $\overline{\alpha}_{x}$ refer respectively to the first passage time of the excursion $\varepsilon$ over $x$ and the excursion height prior to $\rho_{x}^{+}$. The following result gives the discounted joint distribution involving $\rho_{x}^{+}$ and $\overline{\alpha}_{x}$ under the excursion measure. It generalizes Lemma 2.2 of Kyprianou and Zhou (2009) by imposing a constraint on the excursion height
 $\overline{\alpha}_{x}\leq x-v$, where $v\in(0,x]$ can be taken as the minimum capital requirement on the last minimum of the surplus before ruin.

\begin{lem}\label{lemma1}
For $x, y, z\in (0,+\infty)$ and $v\in(0,x]$, we have
\begin{eqnarray}\label{v.plu.}
\hspace{-0.3cm}&&
\hspace{-0.3cm}
n\left(\mathrm{e}^{-q \rho_{x}^{+}}
; x-\varepsilon\left(\rho_{x}^{+}-\right)\in\mathrm{d}y,
\,\varepsilon\left(\rho_{x}^{+}\right)-x\in\mathrm{d}z
,\,\overline{\alpha}_{x}\leq x-v
\right)
\nonumber\\
\hspace{-0.3cm}&=&
\hspace{-0.3cm}
\left(W_{q}^{\prime}(x-y)-W_{q}(x-y)\frac{W_{q}^{\prime}(x-v)}{W_{q}(x-v)}\right)
\upsilon(\mathrm{d}z+y)\,\mathbf{1}_{\{y<x\}}\,\mathrm{d}y
\nonumber\\
\hspace{-0.3cm}&&\hspace{-0.3cm}
+\,W_{q}(0)\,
\upsilon(\mathrm{d}z+x)\,
\delta_{x}(\mathrm{d}y),
\end{eqnarray}
where $\delta_{x}(\mathrm{d}y)$ is the Dirac measure which assigns unit mass to the point $x$.
\end{lem}

\begin{proof}
[Proof\emph{:}]\,\,\,
Recall $\underline{X}(t)=\inf\limits_{0\leq s\leq t}X(s)$ for $t\geq0$, then $\underline{X}(\tau_{0}^{-}-)$ refers to the last minimum of the surplus before ruin, see Biffis and Morales (2010). By Theorem 1 in Biffis and Kyprianou (2010) one has
\begin{eqnarray}\label{}
\hspace{-0.3cm}&&\hspace{-0.3cm}
\mathbb{E}_{x}\left(\mathrm{e}^{-q\tau_{0}^{-}}; X(\tau_{0}^{-}-)\in\mathrm{d}y,
-X(\tau_{0}^{-})\in\mathrm{d}z,
\underline{X}(\tau_{0}^{-}-)\in\mathrm{d}v\right)
\nonumber\\
\hspace{-0.3cm}&=&\hspace{-0.3cm}
\mathrm{e}^{-\Phi_{q}(y-v)}\left(W_{q}^{\prime}(x-v)-\Phi_{q}W_{q}(x-v)\right)
\upsilon(\mathrm{d}z+y)\,\mathrm{d}v\,\mathrm{d}y, \quad x, y, z\in (0,+\infty),v\in(0,x\wedge y],\nonumber
\end{eqnarray}
which yields
\begin{eqnarray}\label{Bf.Ky.}
\hspace{-0.3cm}&&\hspace{-0.3cm}
\mathbb{E}_{x}\left(\mathrm{e}^{-q\tau_{0}^{-}}; X(\tau_{0}^{-}-)\in\mathrm{d}y,
-X(\tau_{0}^{-})\in\mathrm{d}z,
\underline{X}(\tau_{0}^{-}-)\geq v\right)
\nonumber\\
\hspace{-0.3cm}&=&\hspace{-0.3cm}
\int_{w\in[v,x\wedge y]}\mathrm{e}^{\Phi_{q}w}\left(W_{q}^{\prime}(x-w)-\Phi_{q}W_{q}(x-w)\right)\mathrm{d}w\,
\mathrm{e}^{-\Phi_{q}y}\upsilon(\mathrm{d}z+y)\,\mathrm{d}y
\nonumber\\
\hspace{-0.3cm}&=&\hspace{-0.3cm}
\mathrm{e}^{-\Phi_{q}y}\,\upsilon(\mathrm{d}z+y)\,\mathrm{d}y
\,\left(\mathbf{1}_{\{y< x\}}\int_{w\in[v,y]}
\mathrm{d}\left(-\,\mathrm{e}^{\Phi_{q}w}W_{q}(x-w)\right)
+\mathbf{1}_{\{y\geq x\}}\left(\int_{w\in[v,x)}
\mathrm{d}\left(-\,\mathrm{e}^{\Phi_{q}w}W_{q}(x-w)\right)
\right.\right.
\nonumber\\
\hspace{-0.3cm}&&\hspace{-0.3cm}
\left.\left.+\int_{w\in\{x\}}
\mathrm{d}\left(-\,\mathrm{e}^{\Phi_{q}w}W_{q}(x-w)\right)
\right)\right)
\nonumber\\
\hspace{-0.3cm}&=&\hspace{-0.3cm}
\mathrm{e}^{-\Phi_{q}y}\,
\upsilon(\mathrm{d}z+y)\,\mathrm{d}y
\,\left(\mathbf{1}_{\{y< x\}}\,\left(\mathrm{e}^{\Phi_{q}v}W_{q}(x-v)-\mathrm{e}^{\Phi_{q}y}W_{q}(x-y)
\right)\right.
\nonumber\\
\hspace{-0.3cm}&&\hspace{-0.3cm}
\left.+\mathbf{1}_{\{y\geq x\}}\,\left(\mathrm{e}^{\Phi_{q}x}W_{q}(0)+\left(\mathrm{e}^{\Phi_{q}v}W_{q}(x-v)
-\mathrm{e}^{\Phi_{q}x}W_{q}(0)
\right)\right)\right)
,\quad x,y,z\in(0,\infty),v\in(0,x\wedge y].
\end{eqnarray}

By the same language of excursions as in \eqref{renewal.iden.via.n}, we rewrite \eqref{Bf.Ky.} as
\begin{eqnarray}\label{B10.v}
\hspace{-0.3cm}&&\hspace{-0.3cm}
\mathbb{E}_{x}\left(\mathrm{e}^{-q\tau_{0}^{-}}; \,X(\tau_{0}^{-}-)\in\mathrm{d}y,
\,-X(\tau_{0}^{-})\in\mathrm{d}z,\,\underline{X}(\tau_{0}^{-}-)\geq v\right)
\nonumber\\
\hspace{-0.3cm}&=&
\hspace{-0.3cm}
\mathbb{E}_{x}\left(\sum_{\theta\geq0}\mathrm{e}^{-qL^{-1}(\theta-)}
\prod_{t<\theta}\mathbf{1}_{\{\overline{\varepsilon}_{t}< x+t-v\}}
\mathrm{e}^{-q \rho_{ x+\theta}^{+}(\theta)}
\mathbf{1}_{\{\overline{\varepsilon}_{\theta}> x+\theta\}}\right.
\nonumber\\
\hspace{-0.3cm}&&
\hspace{-0.3cm}
 \left.\times\mathbf{1}_{\{x+\theta-\varepsilon_{\theta}\left(\rho_{ x+\theta}^{+}(\theta)-\right)\in\mathrm{d}y,
\,-\left(x+\theta-\varepsilon_{\theta}\left(\rho_{ x+\theta}^{+}(\theta)\right)\right)\in\mathrm{d}z,
\,\overline{\alpha}_{x+\theta}(\theta)\leq x+\theta-v\}}\right)
,\quad x,y,z\in(0,\infty),v\in(0,x\wedge y],
\end{eqnarray}
where $\varepsilon_{\theta}$ represents the last excursion prior to $\tau_{0}^{-}$, $\rho_{ x+\theta}^{+}(\theta)$ and $\overline{\alpha}_{x+\theta}(\theta)$ are defined by (\ref{exc.upcro.b}) and \eqref{exc.height.b} with $\varepsilon$ replaced by $\varepsilon_{\theta}$ and $b$ replaced by $ x+\theta$, and
$$\underline{X}(\tau_{0}^{-}-)=\inf_{0\leq t<\theta}\left(x+t-\overline{\varepsilon}_{t}\right)\wedge \left(x+\theta-\overline{\alpha}_{x+\theta}(\theta)\right).$$
By \eqref{B10.v}, the well known two-sided exit identity (see, Kyprianou (2006)),
$$\mathbb{E}_{x}\left(\mathrm{e}^{-q\tau_{s}^{+}}
\mathbf{1}_{\{\tau_{s}^{+}<\tau_{v}^{-}\}}\right)=\frac{W_{q}(x-v)}{W_{q}(s-v)},\quad 0<v<x<s<\infty,$$
and the compensation formula (cf., Bertoin (1996)) in excursion theory, one can obtain
\begin{eqnarray}\label{B1}
\hspace{-0.3cm}&&\hspace{-0.3cm}
\mathbb{E}_{x}\left(\mathrm{e}^{-q\tau_{0}^{-}}; \,X(\tau_{0}^{-}-)\in\mathrm{d}y,
\,-X(\tau_{0}^{-})\in\mathrm{d}z,\,\underline{X}(\tau_{0}^{-}-)\geq v\right)
\nonumber\\
\hspace{-0.3cm}&=&
\hspace{-0.3cm}
\int_{0}^{\infty}\mathbb{E}_{x}\left(\mathrm{e}^{-qL^{-1}(\theta-)}
\prod_{t<\theta}\mathbf{1}_{\{\overline{\varepsilon}_{t}< x+t-v\}}\right)
\,\times n\left(\mathrm{e}^{-q \rho_{ x+\theta}^{+}(\theta)}
\mathbf{1}_{\{\overline{\varepsilon}_{\theta}> x+\theta\}}\right.
\nonumber\\
\hspace{-0.3cm}&&
\hspace{1.2cm}
\times\left.
 \mathbf{1}_{\{x+\theta-\varepsilon_{\theta}\left(\rho_{ x+\theta}^{+}(\theta)-\right)\in\mathrm{d}y,
\,-\left(x+\theta-\varepsilon_{\theta}\left(\rho_{ x+\theta}^{+}(\theta)\right)\right)\in\mathrm{d}z,\,\overline{\alpha}_{x+\theta}(\theta)\leq x+\theta-v\}}\right)
\mathrm{d}\theta
\nonumber\\
\hspace{-0.3cm}&=&
\hspace{-0.3cm}
\int_{0}^{\infty}\mathbb{E}_{x}\left(\mathrm{e}^{-q\tau_{x+\theta}^{+}}
\mathbf{1}_{\{\tau_{x+\theta}^{+}<\tau_{v}^{-}\}}\right)
\,n\left(\mathrm{e}^{-q \rho_{ x+\theta}^{+}}
\mathbf{1}_{\{\overline{\varepsilon}> x+\theta\}}
 \mathbf{1}_{\{x+\theta-\varepsilon\left(\rho_{ x+\theta}^{+}-\right)\in\mathrm{d}y,
\,\varepsilon\left(\rho_{ x+\theta}^{+}\right)-\left(x+\theta\right)\in\mathrm{d}z
,\,\overline{\alpha}_{x+\theta}\leq x+\theta-v\}}\right)
\mathrm{d}\theta
\nonumber\\
\hspace{-0.3cm}&=&
\hspace{-0.3cm}
\int_{x}^{\infty}\mathbb{E}_{x}\left(\mathrm{e}^{-q\tau_{s}^{+}}
\mathbf{1}_{\{\tau_{s}^{+}<\tau_{v}^{-}\}}\right)
 n\left(\mathrm{e}^{-q \rho_{s}^{+}}
; \overline{\varepsilon}>s, s-\varepsilon\left(\rho_{s}^{+}-\right)\in\mathrm{d}y,
\,\varepsilon\left(\rho_{s}^{+}\right)-s\in\mathrm{d}z
,\,\overline{\alpha}_{s}\leq s-v\right)
\mathrm{d}s
\nonumber\\
\hspace{-0.3cm}&=&
\hspace{-0.3cm}
\int_{x}^{\infty}\frac{W_{q}(x-v)}{W_{q}(s-v)}
 \,n\left(\mathrm{e}^{-q \rho_{s}^{+}}
; s-\varepsilon\left(\rho_{s}^{+}-\right)\in\mathrm{d}y,
\,\varepsilon\left(\rho_{s}^{+}\right)-s\in\mathrm{d}z
,\,\overline{\alpha}_{s}\leq s-v
\right)
\mathrm{d}s
,\nonumber
\end{eqnarray}
which combined with \eqref{Bf.Ky.} yields
\begin{eqnarray}
\hspace{-0.3cm}&&
\hspace{-0.3cm}
\frac{1}{W_{q}(x-v)}\,n\left(\mathrm{e}^{-q \rho_{x}^{+}}
; x-\varepsilon\left(\rho_{x}^{+}-\right)\in\mathrm{d}y,
\,\varepsilon\left(\rho_{x}^{+}\right)-x\in\mathrm{d}z
,\,\overline{\alpha}_{x}\leq x-v
\right)
\nonumber\\
\hspace{-0.3cm}&=&
\hspace{-0.3cm}
\frac{W_{q}^{\prime}(x-y)W_{q}(x-v)-W_{q}(x-y)W_{q}^{\prime}(x-v)}{\left(W_{q}(x-v)\right)^{2}}
\mathbf{1}_{\{x>y\}}
\upsilon(\mathrm{d}z+y)\,\mathrm{d}y
\nonumber\\
\hspace{-0.3cm}&&\hspace{-0.3cm}
+\frac{W_{q}(0)}{W_{q}(x-v)}\,
\upsilon(\mathrm{d}z+x)\,
\delta_{x}(\mathrm{d}y)
,\quad x,y,z\in(0,\infty),v\in(0,x\wedge y],\nonumber
\end{eqnarray}
which is \eqref{v.plu.}. The proof is complete.
\end{proof}

\medskip
The following Lemma \ref{lemma3} solves the general drawdown based two-sided exit problem, and can be found in  Proposition 3.1 of Li et al. (2017) and Lemma 3.2 of Wang and Zhou (2018).

\medskip
\begin{lem}\label{lemma3}
For $q>0$ and $s>x$ and general drawdown function $\xi$, we have
\begin{eqnarray}\label{}
\mathbb{E}_{x}\left(\mathrm{e}^{-q\tau_{s}^{+}}\mathbf{1}_{\{\tau_{s}^{+}<\tau_{\xi}\}}\right)
=\exp\left(-\int_{x}^{s}\frac{W_{q}^{\prime}(\overline{\xi}\left(z\right))}
{W_{q}(\overline{\xi}\left(z\right))}\mathrm{d}z\right),\nonumber
\end{eqnarray}
where $\overline{\xi}\left(z\right)=z-\xi(z)$.
\end{lem}

\medskip
Motivated by the last minimum of the surplus before ruin proposed by Biffis and Morales (2010), this paper includes a constraint on the last minimum surplus before drawdown intending to better reflect the company's financial condition. Let $\vartheta: [x,\infty)\rightarrow[0,\infty)$ be a measurable function satisfying $0\leq\vartheta(z)<\overline{\xi}(z)$. In our main result in below, we use $\vartheta(\overline{X}(t))$ to describe the minimum capital requirement of the surplus above the drawdown level. Note that when $\vartheta(\overline{X}(t))$ equals to a constant $v$ and $\xi(z)\equiv 0$ we degenerate to the case of ruin adopted by Lemma \ref{lemma1}. 
Put $\varsigma(z)=\vartheta(z)+\xi(z)$ and $\overline{\varsigma}(z)=z-(\vartheta(z)+\xi(z))=\overline{\xi}(z)-\vartheta(z)$, hence $\varsigma$ is also a general drawdown function. Theorem \ref{3.1} derives the extended expected discounted penalty function at the general drawdown time for L\'{e}vy risk processes, where a similar excursion approach as in Pistorius (2007), Kyprianou and Zhou (2009) and Li et al (2017) was adopted.



\medskip
\begin{thm}\label{3.1}
Denote by $\upsilon$ the L\'{e}vy measure of $-X$, by $\ell:=L^{-1}(L(\tau_{\xi})-)$ the first time when the L\'{e}vy process $X$ hits the running maximum prior to the general drawdown time.
\begin{itemize}
\item[(a)]
For $s\in(x\vee y,\infty)$, $y\in[\xi(s),\infty)$, $z\in (-\xi(s),\infty)$ and $q,\,\lambda \geq0$, we have
\begin{eqnarray}\label{pen.1}
\hspace{-0.3cm}&&\hspace{-0.3cm}
\mathbb{E}_{x}\left(\mathrm{e}^{-q\ell-\lambda \left(\tau_{\xi}-\ell\right)}; \,\overline{X}(\tau_{\xi})\in \mathrm{d}s,\,X(\tau_{\xi}-)\in\mathrm{d}y,
\,-X(\tau_{\xi})\in\mathrm{d}z
,\inf_{t\in[\,0,\tau_{\xi})}\left(X(t)-\varsigma(\overline{X}(t))\right)\geq 0\right)
\nonumber\\
\hspace{-0.3cm}&=&
\hspace{-0.3cm}
\exp\left(-\int_{x}^{s}\frac{W_{q}^{\prime}(\overline{\varsigma}\left(w\right))}
{W_{q}(\overline{\varsigma}\left(w\right))}\mathrm{d}w\right)
\Bigg(W_{\lambda}(0+)\,\upsilon(s+\mathrm{d}z)\,\delta_{s}(\mathrm{d}y)
\nonumber\\
\hspace{-0.3cm}&&
\hspace{-0.3cm}
\left.
+\left(W_{\lambda}^{\prime}(s-y)-\frac{W_{\lambda}^{\prime}(\overline{\varsigma}(s))}
{W_{\lambda}(\overline{\varsigma}(s))}W_{\lambda}(s-y)\right)
\upsilon\left(y+\mathrm{d}z\right)
\,\mathbf{1}_{\{y<s\}}
\mathrm{d}y\right)
\mathrm{d}s.
\end{eqnarray}
\item[(b)]For $s\in(x,\infty)$ and $q,\,\lambda \geq0$, we have
\begin{eqnarray}\label{pen.2}
\hspace{-0.3cm}&&\hspace{-0.3cm}
\mathbb{E}_{x}\left(\mathrm{e}^{-q\ell-\lambda \left(\tau_{\xi}-\ell\right)}; \,\overline{X}(\tau_{\xi})\in \mathrm{d}s,
\,X(\tau_{\xi})=\xi(s)\right)
\nonumber\\
\hspace{-0.3cm}&=&
\hspace{-0.3cm}
\frac{\sigma^{2}}{2}\exp\left(-\int_{x}^{s}\frac{W_{q}^{\prime}(\overline{\xi}\left(w\right))}
{W_{q}(\overline{\xi}\left(w\right))}\mathrm{d}w\right)
\left(\frac{\left(W_{\lambda}^{\prime}\left(\overline{\xi}\left(s\right)\right)\right)^{2}}
{W_{\lambda}\left(\overline{\xi}\left(s\right)\right)}
-W_{\lambda}^{\prime\prime}\left(\overline{\xi}\left(s\right)\right)\right)
\mathrm{d}s,
\end{eqnarray}
where the expression is understood to be equal to $0$ if $\sigma=0$. At the time of drawdown, the surplus stays at the general drawdown level with a positive probability if and only if the Gaussian part of the L\'{e}vy process is nontrivial.
\end{itemize}
\end{thm}

\medskip
\begin{rem}
Note that on the left hand side of equation (\ref{pen.1}) we include a constraint
$$\inf_{t\in[\,0,\tau_{\xi})}\left(X(t)-\varsigma(\overline{X}(t))\right)\geq 0\Leftrightarrow
X(t)-\xi(\overline{X}(t))\geq \vartheta(\overline{X}(t)),\,\, 0\leq t<\tau_{\xi},$$
which in fact is a constraint on $X(t)-\xi(\overline{X}(t))$, that is, the level of surplus that is above the drawdown level $\xi(\overline{X}(t))$. We require this level to be at least $\vartheta(\overline{X}(t))\geq0$, which could be linked with the confidence level of the company and hence to be dependent with the historical running maximum $\overline{X}(t)$. The lower $\vartheta(\overline{X}(t))$, the worse the financial conditions that need to be negotiated with the company's capital providers, and the more urgent for the company to examine its financial activities. The level of $\vartheta(\overline{X}(t))$ provides a warning line of the company's inadvisable financial decisions such as a low premium rate, also serves as a buffer towards future's capital injections.
\end{rem}

\medskip
\begin{proof}[Proof of Theorem \ref{3.1}\emph{:}]\,\,\,
Similar to the idea in Lemma \ref{lemma2}, the definition of general drawdown leads to the existence of an excursion with a positive lifetime, such that $\tau_{\xi}$ lies in between the left and right end points of this excursion.
Denote this excursion by $\varepsilon_{\theta}$ with $\theta\geq0$, then $\varepsilon_{\theta}$ is the last excursion whose left-end point is less than $\tau_{\xi}$, and we have $\overline{\varepsilon}_{\theta}>\overline{\xi}(x+\theta)$ and $\overline{\varepsilon}_{t}\leq\overline{\xi}(x+t)$ for all excursions $\varepsilon_{t}$ with $t<\theta$. That is, a general drawdown does not occur during the lifetime of $\varepsilon_{t}$ with $t<\theta$, while it does occur during the lifetime of $\varepsilon_{\theta}$.
Furthermore, one can translate immediately the surplus before and at the general drawdown time by the excursion $\varepsilon_{\theta}$ through
$$X(\tau_{\xi}-)=x+\theta-\varepsilon_{\theta}\left(\rho_{\overline{\xi}(x+\theta)}^{+}(\theta)-\right)
,\quad X(\tau_{\xi})=x+\theta-\varepsilon_{\theta}\left(\rho_{\overline{\xi}(x+\theta)}^{+}(\theta)\right),$$
where $\rho_{\overline{\xi}(x+\theta)}^{+}(\theta)$ is defined via (\ref{exc.upcro.b}) with $\varepsilon$ replaced by $\varepsilon_{\theta}$ and $b$ replaced by $\overline{\xi}(x+\theta)$.
Meanwhile, the general drawdown time can be rewritten as
$$\tau_{\xi}=L^{-1}(\theta-)+\rho_{\overline{\xi}(x+\theta)}^{+}(\theta).$$
Therefore we have that, for $y, z\in (-\infty,+\infty)$, $q,\,\lambda \geq0$ and any open interval $B\subseteq(x,\infty)$,
\begin{eqnarray}\label{B10}
\hspace{-0.3cm}&&\hspace{-0.3cm}
\mathbb{E}_{x}\left(\mathrm{e}^{-q\ell-\lambda \left(\tau_{\xi}-\ell\right)}; \,\overline{X}(\tau_{\xi})\in B,\,X(\tau_{\xi}-)\in\mathrm{d}y,
\,-X(\tau_{\xi})\in\mathrm{d}z,\inf_{t\in[0,\tau_{\xi})}\left(X(t)-\varsigma(\overline{X}(t))\right)\geq 0\right)
\nonumber\\
\hspace{-0.3cm}&=&
\hspace{-0.3cm}
\mathbb{E}_{x}\left(\sum_{\theta\in B-x}\mathrm{e}^{-qL^{-1}(\theta-)}
\prod_{t<\theta}\mathbf{1}_{\{\overline{\varepsilon}_{t}<\overline{\varsigma}\left(x+t\right)\}}
\mathrm{e}^{-\lambda \rho_{\overline{\xi}(x+\theta)}^{+}(\theta)}
\mathbf{1}_{\{\overline{\varepsilon}_{\theta}>\overline{\xi}(x+\theta)\}}\right.
\nonumber\\
\hspace{-0.3cm}&&
\hspace{1.2cm}
\times\left.
 \mathbf{1}_{\{x+\theta-\varepsilon_{\theta}\left(\rho_{\overline{\xi}(x+\theta)}^{+}(\theta)-\right)\in\mathrm{d}y,
\,-\left(x+\theta-\varepsilon_{\theta}\left(\rho_{\overline{\xi}(x+\theta)}^{+}(\theta)\right)\right)
\in\mathrm{d}z,\,\overline{\alpha}_{\overline{\xi}(x+\theta)}\leq \overline{\varsigma}(x+\theta)\}}\right)
,
\end{eqnarray}
where $\varepsilon_{\theta}$ represents the last excursion prior to $\tau_{\xi}$, $\rho_{\overline{\xi}(x+\theta)}^{+}(\theta)$ is given by (\ref{exc.upcro.b}) with $\varepsilon$ replaced by $\varepsilon_{\theta}$ and $b$ replaced by $\overline{\xi}(x+\theta)$, and $B-x:=\left.\{y-x\right| y\in B\}$.

Using the compensation formula (see, Corollary 11 in Chapter IV.4 of Bertoin (1996)) in the excursion expression in (\ref{B10}), we obtain
\begin{eqnarray}\label{B1}
\hspace{-0.3cm}&&\hspace{-0.3cm}
\mathbb{E}_{x}\left(\mathrm{e}^{-q\ell-\lambda \left(\tau_{\xi}-\ell\right)}; \,\overline{X}(\tau_{\xi})\in B,\,X(\tau_{\xi}-)\in\mathrm{d}y,
\,-X(\tau_{\xi})\in\mathrm{d}z,\inf_{t\in[0,\tau_{\xi})}\left(X(t)-\varsigma(\overline{X}(t))\right)\geq 0\right)
\nonumber\\
\hspace{-0.3cm}&=&
\hspace{-0.3cm}
\int_{B-x}\mathbb{E}_{x}\left(\mathrm{e}^{-qL^{-1}(\theta-)}
\prod_{t<\theta}\mathbf{1}_{\{\overline{\varepsilon}_{t}<\overline{\varsigma}\left(x+t\right)\}}\right)
\,\times n\left(\mathrm{e}^{-\lambda \rho_{\overline{\xi}(x+\theta)}^{+}(\theta)}
\mathbf{1}_{\{\overline{\varepsilon}_{\theta}>\overline{\xi}(x+\theta)\}}\right.
\nonumber\\
\hspace{-0.3cm}&&
\hspace{1.2cm}
\times\left.
 \mathbf{1}_{\{x+\theta-\varepsilon_{\theta}\left(\rho_{\overline{\xi}(x+\theta)}^{+}(\theta)-\right)
 \in\mathrm{d}y,
\,-\left(x+\theta-\varepsilon_{\theta}
\left(\rho_{\overline{\xi}(x+\theta)}^{+}(\theta)\right)\right)\in\mathrm{d}z
,\,\overline{\alpha}_{\overline{\xi}(x+\theta)}\leq \overline{\varsigma}(x+\theta)\}}\right)
\mathrm{d}\theta
\nonumber\\
\hspace{-0.3cm}&=&
\hspace{-0.3cm}
\int_{B-x}\mathbb{E}_{x}\left(\mathrm{e}^{-q\tau_{x+\theta}^{+}}
\mathbf{1}_{\{\tau_{x+\theta}^{+}<\tau_{\varsigma}\}}\right)
\,\times n\left(\mathrm{e}^{-\lambda \rho_{\overline{\xi}(x+\theta)}^{+}}
\mathbf{1}_{\{\overline{\varepsilon}>\overline{\xi}(x+\theta)\}}\right.
\nonumber\\
\hspace{-0.3cm}&&
\hspace{1.2cm}
\times\left.
 \mathbf{1}_{\{x+\theta-\varepsilon\left(\rho_{\overline{\xi}(x+\theta)}^{+}-\right)\in\mathrm{d}y,
\,\varepsilon\left(\rho_{\overline{\xi}(x+\theta)}^{+}\right)-\left(x+\theta\right)\in\mathrm{d}z
,\,\overline{\alpha}_{\overline{\xi}(x+\theta)}\leq \overline{\varsigma}(x+\theta)\}}\right)
\mathrm{d}\theta
\nonumber\\
\hspace{-0.3cm}&=&
\hspace{-0.3cm}
\int_{B}\mathbb{E}_{x}\left(\mathrm{e}^{-q\tau_{s}^{+}}
\mathbf{1}_{\{\tau_{s}^{+}<\tau_{\varsigma}\}}\right)
 n\left(\mathrm{e}^{-\lambda \rho_{\overline{\xi}(s)}^{+}}
; \overline{\varepsilon}>\overline{\xi}(s), s-\varepsilon\left(\rho_{\overline{\xi}(s)}^{+}-\right)\in\mathrm{d}y,
\,\varepsilon\left(\rho_{\overline{\xi}(s)}^{+}\right)-s\in\mathrm{d}z
,\,\overline{\alpha}_{\overline{\xi}(s)}\leq \overline{\varsigma}(s)\right)
\mathrm{d}s
\nonumber\\
\hspace{-0.3cm}&=&
\hspace{-0.3cm}
\int_{B}\exp\left(-\int_{x}^{s}\frac{W_{q}^{\prime}(\overline{\varsigma}\left(w\right))}
{W_{q}(\overline{\varsigma}\left(w\right))}\mathrm{d}w\right)\,\mathbf{1}_{\{y\geq \xi(s)\}}\,\mathbf{1}_{\{z>-\xi(s)\}}
\nonumber\\
\hspace{-0.3cm}&&
\hspace{-0.3cm}
\times\, n\left(\mathrm{e}^{-\lambda \rho_{\overline{\xi}(s)}^{+}}
; \overline{\xi}(s)-\varepsilon\left(\rho_{\overline{\xi}(s)}^{+}-\right)\in-\xi(s)+\mathrm{d}y,
\,\varepsilon\left(\rho_{\overline{\xi}(s)}^{+}\right)-\overline{\xi}(s)\in\xi(s)+\mathrm{d}z
,\,\overline{\alpha}_{\overline{\xi}(s)}\leq \overline{\varsigma}(s)\right)
\mathrm{d}s
\nonumber\\
\hspace{-0.3cm}&=&
\hspace{-0.3cm}
\int_{B}\exp\left(-\int_{x}^{s}\frac{W_{q}^{\prime}(\overline{\varsigma}\left(w\right))}
{W_{q}(\overline{\varsigma}\left(w\right))}\mathrm{d}w\right)\,\mathbf{1}_{\{y\geq \xi(s),\,z>-\xi(s)\}}
\Bigg(W_{\lambda}(0+)\,\upsilon(s+\mathrm{d}z)\,\delta_{s}(\mathrm{d}y)
\nonumber\\
\hspace{-0.3cm}&&
\hspace{-0.3cm}
\left.
+\left(W_{\lambda}^{\prime}(s-y)-\frac{W_{\lambda}^{\prime}(\overline{\varsigma}(s))}
{W_{\lambda}(\overline{\varsigma}(s))}W_{\lambda}(s-y)\right)
\upsilon\left(y+\mathrm{d}z\right)
\,\mathbf{1}_{\{y<s\}}
\mathrm{d}y\right)
\mathrm{d}s
,
\end{eqnarray}
where in the last equality we have used Lemma \ref{lemma1}, and in the last but one equality we have used Lemma \ref{lemma3}.
The arbitrariness of the open interval $B$ together with (\ref{B1}) yields
(\ref{pen.1}). 

It remains to prove Case $(b)$. In fact, by the compensation formula (see, Corollary 11 in Chapter IV.4 of Bertoin (1996)) we have
\begin{eqnarray}\label{pen.2.0}
\hspace{-0.3cm}&&\hspace{-0.3cm}
\mathbb{E}_{x}\left(\mathrm{e}^{-q\ell-\lambda \left(\tau_{\xi}-\ell\right)}; \,\overline{X}(\tau_{\xi})\in B,
\,X(\tau_{\xi})=\xi\left(\overline{X}(\tau_{\xi})\right)\right)
\nonumber\\
\hspace{-0.3cm}&=&
\hspace{-0.3cm}
\mathbb{E}_{x}\left(\sum_{\theta\in B-x}\mathrm{e}^{-qL^{-1}(\theta-)}
\prod_{t<\theta}\mathbf{1}_{\{\overline{\varepsilon}_{t}<\overline{\xi}(x+t)\}}
\mathrm{e}^{-\lambda \rho_{\overline{\xi}(x+\theta)}^{+}(\theta)}
\mathbf{1}_{\{\overline{\varepsilon}_{\theta}>\overline{\xi}(x+\theta)\}}
 \mathbf{1}_{\{\varepsilon_{\theta}(\rho_{\overline{\xi}(x+\theta)}^{+}(\theta))
=\overline{\xi}\left(x+\theta\right)\}}\right)
\nonumber\\
\hspace{-0.3cm}&=&
\hspace{-0.3cm}
\int_{B-x}\mathbb{E}_{x}\left(\mathrm{e}^{-qL^{-1}(\theta-)}
\prod_{t<\theta}\mathbf{1}_{\{\overline{\varepsilon}_{t}<\overline{\xi}(x+t)\}}\right)
\nonumber\\
\hspace{-0.3cm}&&
\hspace{0.6cm}
 \times
 n\left(\mathrm{e}^{-\lambda \rho_{\overline{\xi}(x+\theta)}^{+}(\theta)}
\mathbf{1}_{\{\overline{\varepsilon}_{\theta}>\overline{\xi}(x+\theta)\}}
\mathbf{1}_{\{\varepsilon_{\theta}(\rho_{\overline{\xi}(x+\theta)}^{+}(\theta))
=\overline{\xi}\left(x+\theta\right)\}}\right)
\mathrm{d}\theta
\nonumber\\
\hspace{-0.3cm}&=&
\hspace{-0.3cm}
\int_{B-x}\mathbb{E}_{x}\left(\mathrm{e}^{-q\tau_{x+\theta}^{+}}
\mathbf{1}_{\{\tau_{x+\theta}^{+}<\tau_{\xi}\}}\right)
 n\left(\mathrm{e}^{-\lambda \rho_{\overline{\xi}(x+\theta)}^{+}}
\mathbf{1}_{\{\overline{\varepsilon}>\overline{\xi}(x+\theta)\}}
\mathbf{1}_{\{\varepsilon(\rho_{\overline{\xi}(x+\theta)}^{+})
=\overline{\xi}\left(x+\theta\right)\}}\right)
\mathrm{d}\theta
\nonumber\\
\hspace{-0.3cm}&=&
\hspace{-0.3cm}
\int_{B}\exp\left(-\int_{x}^{s}\frac{W_{q}^{\prime}(\overline{\xi}\left(w\right))}
{W_{q}(\overline{\xi}\left(w\right))}\mathrm{d}w\right)
\, n\left(\mathrm{e}^{-\lambda \rho_{\overline{\xi}(s)}^{+}}
; \overline{\varepsilon}>\overline{\xi}(s),\,\varepsilon\left(\rho_{\overline{\xi}(s)}^{+}\right)=\overline{\xi}(s)\right)
\mathrm{d}s
\nonumber\\
\hspace{-0.3cm}&=&
\hspace{-0.3cm}
\int_{B}\exp\left(-\int_{x}^{s}\frac{W_{q}^{\prime}(\overline{\xi}\left(w\right))}
{W_{q}(\overline{\xi}\left(w\right))}\mathrm{d}w\right)
\frac{\sigma^{2}}{2}\left(\frac{\left(W_{\lambda}^{\prime}\left(\overline{\xi}\left(s\right)\right)\right)^{2}}
{W_{\lambda}\left(\overline{\xi}\left(s\right)\right)}
-W_{\lambda}^{\prime\prime}\left(\overline{\xi}\left(s\right)\right)\right)
\mathrm{d}s
,
\end{eqnarray}
where in the last equality we have used Lemma \ref{lemma2}, and in the last but one equality we have used Lemma \ref{lemma3}.
The arbitrariness of the open interval $B$ together with (\ref{pen.2.0}) yields
(\ref{pen.2}). This completes the proof of Case ($b$) of Theorem \ref{3.1}.
\end{proof}

\medskip
\begin{rem}
It should be pointed out that the excursion methodology adopted here to study the general drawdown is based on L\'{e}vy processes, and hence can not be easily applied to other underlying models, say for example, diffusion processes or jump diffusion processes. In Landriault et al. (2017), a methodology called ``short-time pathwise analysis" was proposed to study the classical drawdown involved fluctuation problems for general time-homogeneous Markov processes. Their method is supposed to be more general in the sense that can be adapted to study fluctuation problems involving the general drawdown time, but still seems restrictive in our case where both the time to reach the running maximum at $\tau_{\xi}$ and the surplus level before $\tau_{\xi}$ are involved.
\end{rem}

\medskip
Note that in Part $(a)$ of Theorem 3.1, $y\in[\xi(s),\infty)$ is equivalent to $y-\xi(s)\in[0,\infty)$, and $z\in(-\xi(s),\infty)$ is equivalent to $z+\xi(s)\in(0,\infty)$, we readily have the following result.
\begin{cor}
For $z>0$, $y\geq0$, $s\geq y\vee x$, and $q,\,\lambda \geq0$, we have
\begin{eqnarray}\label{}
\hspace{-0.3cm}&&\hspace{-0.3cm}
\mathbb{E}_{x}\left(\mathrm{e}^{-q\ell-\lambda \left(\tau_{\xi}-\ell\right)}; \,\overline{X}(\tau_{\xi})\in\mathrm{d}s,\,X(\tau_{\xi}-)-\xi(s)\in\mathrm{d}y,
\,\xi(s)-X(\tau_{\xi})\in\mathrm{d}z
,\inf_{t\in[\,0,\tau_{\xi})}\left(X(t)-\varsigma(\overline{X}(t))\right)\geq 0\right)
\nonumber\\
\hspace{-0.3cm}&=&
\hspace{-0.3cm}
\exp\left(-\int_{x}^{s}\frac{W_{q}^{\prime}(\overline{\varsigma}\left(w\right))}
{W_{q}(\overline{\varsigma}\left(w\right))}\mathrm{d}w\right)
\Bigg(W_{\lambda}(0+)\,\upsilon(\overline{\xi}(s)+\mathrm{d}z)\,\delta_{\overline{\xi}(s)}(\mathrm{d}y)
\nonumber\\
\hspace{-0.3cm}&&
\hspace{-0.3cm}
\left.
+\left(W_{\lambda}^{\prime}(\overline{\xi}(s)-y)-\frac{W_{\lambda}^{\prime}(\overline{\varsigma}(s))}
{W_{\lambda}(\overline{\varsigma}(s))}W_{\lambda}(\overline{\xi}(s)-y)\right)
\upsilon\left(y+\mathrm{d}z\right)
\,\mathbf{1}_{\{y<\overline{\xi}(s)\}}
\mathrm{d}y\right)\mathrm{d}s
.\nonumber
\end{eqnarray}
Furthermore, when $q=\lambda$, $\xi\equiv0$ and $\vartheta\equiv v\in[0,x)$ we have
\begin{eqnarray}\label{}
\hspace{-0.3cm}&&\hspace{-0.3cm}
\mathbb{E}_{x}\left(\mathrm{e}^{-q\tau_{0}^{-}}; \,\overline{X}(\tau_{0}^{-})\in\mathrm{d}s,\,X(\tau_{0}^{-}-)\in\mathrm{d}y,
\,-X(\tau_{0}^{-})\in\mathrm{d}z
, \,\underline{X}(\tau_{0}^{-})\geq v\right)
\nonumber\\
\hspace{-0.3cm}&=&
\hspace{-0.3cm}
\left.\frac{W_{q}(x-v)}
{W_{q}(s-v)}\,
\right(W_{q}(0+)\,\upsilon(s+\mathrm{d}z)\,\delta_{s}(\mathrm{d}y)
\nonumber\\
\hspace{-0.3cm}&&
\hspace{-0.3cm}
\left.
+\left(W_{q}^{\prime}(s-y)-\frac{W_{q}^{\prime}(s-v)}
{W_{q}(s-v)}W_{q}(s-y)\right)
\upsilon\left(y+\mathrm{d}z\right)
\,\mathbf{1}_{\{y<s\}}
\mathrm{d}y\right)\mathrm{d}s
,\nonumber
\end{eqnarray}
integrating which with respect to $s$ over $[x,\infty)$ and recalling that $\lim\limits_{x\rightarrow\infty}\frac{W_{q}(s-y)}{W_{q}(s-v)}=\mathrm{e}^{-\Phi_{q}(y-v)}$, one can recover \eqref{Bf.Ky.}. Hence, as a special case, the result of Theorem \ref{3.1} coincides with that of Theorem 1 in Biffis and Kyprianou (2010).
\end{cor}

\medskip
When $\vartheta\equiv\xi\equiv0$, then the general drawdown time is reduced to the classical ruin time $\tau_{\xi}=\tau_{0}^{-}$, $\ell=L^{-1}(L(\tau_{0}^{-})-)$, and the results of Theorem \ref{3.1} are specialized to the following Corollary \ref{3.2}, which coincides well with Theorem 1.3 in Kyprianou and Zhou (2009) with $\gamma\equiv0$.

\begin{cor}\label{3.2}
The generalized expected discounted penalty function at the classical ruin time can be characterized as follows.
\begin{itemize}
\item[(a$^{\prime}$)]
For $s\in(x\vee y,\infty)$, $y\in[0,\infty)$, $z\in (0,\infty)$ and $q,\,\lambda \geq0$, we have
\begin{eqnarray}\label{ }
\hspace{-0.3cm}&&\hspace{-0.3cm}
\mathbb{E}_{x}\left(\mathrm{e}^{-q\ell-\lambda \left(\tau_{0}^{-}-\ell\right)}; \,\overline{X}(\tau_{0}^{-})\in \mathrm{d}s,\,X(\tau_{0}^{-}-)\in\mathrm{d}y,
\,-X(\tau_{0}^{-})\in\mathrm{d}z\right)
\nonumber\\
\hspace{-0.3cm}&=&
\hspace{-0.3cm}
\frac{W_{q}(x)}
{W_{q}(s)}
\Bigg(W_{\lambda}(0+)\,\upsilon(s+\mathrm{d}z)\,\delta_{s}(\mathrm{d}y)
\nonumber\\
\hspace{-0.3cm}&&
\hspace{-0.3cm}
+\left(W_{\lambda}^{\prime}(s-y)-\frac{W_{\lambda}^{\prime}(s)}{W_{\lambda}(s)}W_{\lambda}(s-y)\right)
\upsilon\left(y+\mathrm{d}z\right)
\,\mathbf{1}_{\{y<s\}}
\mathrm{d}y\Bigg)
\mathrm{d}s.\nonumber
\end{eqnarray}

\item[(b$^{\prime}$)]For $s\in(x\vee y,\infty)$ and $q,\,\lambda \geq0$, we have
\begin{eqnarray}\label{ }
\hspace{-0.3cm}&&\hspace{-0.3cm}
\mathbb{E}_{x}\left(\mathrm{e}^{-q\ell-\lambda \left(\tau_{0}^{-}-\ell\right)}; \,\overline{X}(\tau_{0}^{-})\in \mathrm{d}s,
\,X(\tau_{0}^{-})=0\right)
=
\frac{\sigma^{2}}{2}\frac{W_{q}(x)}
{W_{q}(s)}
\left(\frac{\left(W_{\lambda}^{\prime}\left(s\right)\right)^{2}}
{W_{\lambda}\left(s\right)}
-W_{\lambda}^{\prime\prime}\left(s\right)\right)
\mathrm{d}s,\nonumber
\end{eqnarray}
where the expression is understood to be equal to $0$ if $\sigma=0$.
\end{itemize}
\end{cor}


\medskip
\begin{rem}
Furthermore, if there is no Brownian part in the L\'{e}vy-It\^{o} decomposition of $X$ (i.e., $\sigma=0$), then by $(a^{\prime})$ one may find, for $b>x$, $y\in[0,\infty)$, $z\in (0,\infty)$ and $q \geq0$
\begin{eqnarray}\label{reduce.}
\hspace{-0.3cm}&&\hspace{-0.3cm}
\mathbb{E}_{x}\left(\mathrm{e}^{-q\tau_{0}^{-}}; \,X(\tau_{0}^{-}-)\in\mathrm{d}y,
\,-X(\tau_{0}^{-})\in\mathrm{d}z,\,\tau_{0}^{-}<\tau_{b}^{+}\right)
\nonumber\\
\hspace{-0.3cm}&=&
\hspace{-0.3cm}
\int_{s\in(x,b)}\frac{W_{q}(x)}
{W_{q}(s)}
\Bigg(W_{q}(0+)\,\upsilon(s+\mathrm{d}z)\,\delta_{s}(\mathrm{d}y)
\nonumber\\
\hspace{-0.3cm}&&
\hspace{-0.3cm}
+\left(W_{q}^{\prime}(s-y)-\frac{W_{q}^{\prime}(s)}{W_{q}(s)}W_{q}(s-y)\right)
\upsilon\left(y+\mathrm{d}z\right)
\,\mathbf{1}_{\{y<s\}}
\mathrm{d}y\Bigg)
\mathrm{d}s
\nonumber\\
\hspace{-0.3cm}&=&
\hspace{-0.3cm}
\frac{W_{q}(x)}
{W_{q}(y)}
W_{q}(0+)\,\upsilon(y+\mathrm{d}z)
\,\mathbf{1}_{\{x<y\}}\mathrm{d}y
\nonumber\\
\hspace{-0.3cm}&&
\hspace{-0.3cm}
+
\int_{s\in(x,b)}\frac{W_{q}(x)}
{W_{q}(s)}\left(W_{q}^{\prime}(s-y)-\frac{W_{q}^{\prime}(s)}{W_{q}(s)}W_{q}(s-y)\right)
\upsilon\left(y+\mathrm{d}z\right)
\left(\mathbf{1}_{\{y\leq x\}}+\mathbf{1}_{\{x<y<s\}}\right)
\mathrm{d}y
\mathrm{d}s
\nonumber\\
\hspace{-0.3cm}&=&
\hspace{-0.3cm}
\frac{W_{q}(x)}
{W_{q}(y)}
W_{q}(0+)\,\upsilon(y+\mathrm{d}z)
\,\mathbf{1}_{\{x<y\}}\mathrm{d}y
\nonumber\\
\hspace{-0.3cm}&&
\hspace{-0.3cm}
+
\int_{s\in(x,b)}\frac{W_{q}(x)}
{W_{q}(s)}\left(W_{q}^{\prime}(s-y)-\frac{W_{q}^{\prime}(s)}{W_{q}(s)}W_{q}(s-y)\right)
\mathrm{d}s\,\upsilon\left(y+\mathrm{d}z\right)
\mathbf{1}_{\{y\leq x\}}
\mathrm{d}y
\nonumber\\
\hspace{-0.3cm}&&
\hspace{-0.3cm}
+
\int_{s\in(y,b)}\frac{W_{q}(x)}
{W_{q}(s)}\left(W_{q}^{\prime}(s-y)-\frac{W_{q}^{\prime}(s)}{W_{q}(s)}W_{q}(s-y)\right)
\mathrm{d}s
\,\upsilon\left(y+\mathrm{d}z\right)
\mathbf{1}_{\{x<y\}}\mathrm{d}y
\nonumber\\
\hspace{-0.3cm}&=&
\hspace{-0.3cm}
\frac{W_{q}(x)}
{W_{q}(y)}
W_{q}(0+)\,\upsilon(y+\mathrm{d}z)
\,\mathbf{1}_{\{x<y\}}\mathrm{d}y
\nonumber\\
\hspace{-0.3cm}&&
\hspace{-0.3cm}
+
W_{q}(x)\,\upsilon\left(y+\mathrm{d}z\right)
\left(
\left(\left.\frac{W_{q}(s-y)}{W_{q}(s)}\right|_{s=x}^{b}\right)
\mathbf{1}_{\{y\leq x\}}
\mathrm{d}y
+
\left(\left.\frac{W_{q}(s-y)}{W_{q}(s)}\right|_{s=y}^{b}\right)
\mathbf{1}_{\{x<y\}}\mathrm{d}y
\right)
\nonumber\\
\hspace{-0.3cm}&=&
\hspace{-0.3cm}
W_{q}(x)\,\upsilon\left(y+\mathrm{d}z\right)
\left(
\frac{W_{q}(b-y)}{W_{q}(b)}-\frac{W_{q}(x-y)}{W_{q}(x)}
\right)\mathrm{d}y
.
\end{eqnarray}
where we have used the fact that $\frac{W_{q}^{\prime}(s-y)}{W_{q}(s)}-\frac{W_{q}^{\prime}(s)}{[W_{q}(s)]^{2}}W_{q}(s-y)
=\frac{\mathrm{d}}{\mathrm{d}s}\left(\frac{W_{q}(s-y)}{W_{q}(s)}\right)$ in the fourth equation, and the fact that
$W_{q}(x-y)=0$ for $y >x$ in the last equation. One can find that (\ref{reduce.}) coincides well with Theorem 5.5 on page 41 of Kyprianou (2013), noting that we have $\upsilon\left(y+\mathrm{d}z\right)=\lambda F\left(y+\mathrm{d}z\right)$ in the classical Cram\'{e}r-Lundberg risk process with claim distribution $F$.
\end{rem}

\medskip
\begin{rem}
We say that the results in Theorem \ref{3.1} is the generalized version of the classical expected penalty function at the general drawdown time. 
 In fact, the classical expected penalty function at the general drawdown time can be written  as
\begin{eqnarray}\label{}
\phi(x):=
\mathbb{E}_{x}\left(\mathrm{e}^{-q\tau_{\xi}}\,\omega\left(X(\tau_{\xi}-),
X(\tau_{\xi})\right)\,\mathbf{1}_{\{\tau_{\xi}<\infty\}}\right),\nonumber
\end{eqnarray}
for some bounded measurable bivariate function $\omega(\cdot,\cdot):\,(-\infty,+\infty)^{2}\rightarrow (0,\infty)$. In actuarial sciences, $\omega$ is called the penalty function (cf., Gerber and Shiu (1998)).
Using Theorem \ref{3.1}, one can solve the classical expected penalty function at the general drawdown as
\begin{eqnarray}\label{GS.f.}
\phi(x)
\hspace{-0.3cm}&=&
\hspace{-0.3cm}
\int_{s\in(x,\infty)}\int_{y\in[\xi(s),s]}\int_{z\in(-\xi(s),\infty)}
\omega(y,-z)\exp\left(-\int_{x}^{s}\frac{W_{q}^{\prime}(\overline{\xi}\left(w\right))}
{W_{q}(\overline{\xi}\left(w\right))}\mathrm{d}w\right)
\Bigg(W_{q}(0+)\,\upsilon(s+\mathrm{d}z)\,\delta_{s}(\mathrm{d}y)
\nonumber\\
\hspace{-0.3cm}&&
\hspace{0.3cm}
\left.
+\left(W_{q}^{\prime}(s-y)-\frac{W_{q}^{\prime}(\overline{\xi}(s))}{W_{q}(\overline{\xi}(s))}W_{q}(s-y)\right)
\upsilon\left(y+\mathrm{d}z\right)
\,\mathbf{1}_{\{y<s\}}
\mathrm{d}y\right)
\mathrm{d}s
\nonumber\\
\hspace{-0.3cm}&&
\hspace{-0.3cm}
+
\frac{\sigma^{2}}{2}\int_{s\in(x,\infty)}\omega\left(\xi(s),\xi(s)\right)\exp\left(-\int_{x}^{s}\frac{W_{q}^{\prime}(\overline{\xi}\left(w\right))}
{W_{q}(\overline{\xi}\left(w\right))}\mathrm{d}w\right)
\left(\frac{\left(W_{q}^{\prime}\left(\overline{\xi}\left(s\right)\right)\right)^{2}}
{W_{q}\left(\overline{\xi}\left(s\right)\right)}
-W_{q}^{\prime\prime}\left(\overline{\xi}\left(s\right)\right)\right)
\mathrm{d}s.
\end{eqnarray}

The following interesting example of penalty function can be found in Gerber and Shiu (1998),
$$
\omega(x,y)=\max\left(K-\mathrm{e}^{a-y},0\right).$$
In this case, $\phi(x)$ is the payoff of a perpetual American
put option with $K$ as the exercise price and $\mathrm{e}^{a}$ be the value of an option-exercise
boundary.
\end{rem}

\medskip
\section{Applications}

This section is focused on the applications of Theorem \ref{3.1}. By specifying the general drawdown function, the methodology in this paper can be adapted naturally to recover results in the literature and to obtain new results for the Gerber-Shiu function at ruin for risk processes embedded with a loss-carry-forward taxation system or a barrier dividend strategy. In fact, it was Landriault et al. (2017) and Li (2015) who first pointed out that ruin problems in loss-carry-forward taxation (resp, De Finetti's dividend) models can be transformed to general drawdown problems for the classical models without taxation (resp, dividend). However, they proposed the idea without the implementation of a particular ruin problem. In addition, the drawdown function studied in Landriault et al. (2017) and Li (2015) was the classical drawdown function in the form of $\xi(x)=x-d$ with $d>0$ (not general drawdown function), although it is foreseen that their approach allows an extension to general drawdown function.

\medskip
\subsection{Application for L\'{e}vy risk processes with a general loss-carry-forward system}

\medskip
In Kyprianou and Zhou (2009), a L\'{e}vy risk model with a general loss-carry-forward tax structure was first considered
\begin{eqnarray}\label{2U}
U_{\gamma}(t):=X(t)-\int_{0}^{t}\gamma(\overline{X}(s))\mathrm{d}\overline{X}(s)
=X(t)-\int_{x}^{\overline{X}(t)}\gamma(w)\mathrm{d}w,
\end{eqnarray}
where $\gamma: [0,+\infty)\rightarrow[0,1)$ is measurable and
$
\int_{0}^{\infty}\left(1-\gamma(w)\right)\mathrm{d}w=\infty.
$
In this formulation, taxes are paid whenever the company is in a profitable situation, defined as being at a running maximum of the surplus process.

We claim that the version of Gerber-Shiu function at ruin obtained in Kyprianou and Zhou (2009) can be recovered by specifying a special drawdown function $\xi$ in Theorem \ref{3.1}.
To this purpose, for $x\in(0,\infty)$, let
\begin{eqnarray}\label{taxdd}
\xi_{\gamma}(z):=\int_{x}^{z}\gamma(w)\mathrm{d}w,\quad z\in[x,\infty),
\end{eqnarray}
which is indeed a drawdown function. One can make the following three observations.
\begin{itemize}
\item[($i$)]
The $\xi_{\gamma}$-drawdown time of $X$ coincides with the ruin time of the taxed risk process \eqref{2U}
$$\tau_{\xi_{\gamma}}=\inf\{t\geq0; X(t)<\xi_{\gamma}(\overline{X}(t))\}=\inf\{t\geq0; U_{\gamma}(t)<0\}:=\tau_{0}^{-}(\gamma),$$
and hence $\ell=L^{-1}(L(\tau_{0}^{-}(\gamma))-)$, i.e., the last moment that tax is paid before the ruin of \eqref{2U}.
\item[($ii$)]
The running supremum process $\{\overline{U_{\gamma}}(t):=\sup_{0\leq s\leq t}U_{\gamma}(s);t\geq0\}$ can be rewritten as
$$
\overline{U_{\gamma}}(t)
=x+\int_{0}^{t}(1-\gamma(\overline{X}(s)))\mathrm{d}\overline{X}(s)
=\overline{\xi_{\gamma}}(\overline{X}(t)),$$
and hence, $
\overline{U_{\gamma}}(\tau_{0}^{-}(\gamma))
=\overline{\xi_{\gamma}}(\overline{X}(\tau_{\xi_{\gamma}}))$ and
$$\overline{U_{\gamma}}(\tau_{0}^{-}(\gamma))\in(s,s+\triangle s)
\Leftrightarrow \overline{X}(\tau_{\xi_{\gamma}})\in
\left(\left(\overline{\xi_{\gamma}}\right)^{-1}(s),\left(\overline{\xi_{\gamma}}\right)^{-1}(s+\triangle s)\right),\quad s\geq x,\,\triangle s>0,$$
with $\left(\overline{\xi_{\gamma}}\right)^{-1}$ being the well-defined inverse function of $\overline{\xi_{\gamma}}$.
\item[($iii$)] The taxed surplus level at and immediately before the ruin time $\tau_{0}^{-}(\gamma)$ are rewritten as
$$U_{\gamma}(\tau_{0}^{-}(\gamma))
=X(\tau_{\xi_{\gamma}})-\xi_{\gamma}(\overline{X}(\tau_{\xi_{\gamma}})),
\quad U_{\gamma}(\tau_{0}^{-}(\gamma)-)
=X(\tau_{\xi_{\gamma}}-)-\xi_{\gamma}(\overline{X}(\tau_{\xi_{\gamma}})),$$
and hence, we have, for $z\geq0$ and $\triangle z>0$
$$-U_{\gamma}(\tau_{0}^{-}(\gamma))\in(z,z+\triangle z)
\Leftrightarrow
-X(\tau_{\xi_{\gamma}})\in\left(-\xi_{\gamma}(\overline{X}(\tau_{\xi_{\gamma}}))+z,\,
-\xi_{\gamma}(\overline{X}(\tau_{\xi_{\gamma}}))+z+\triangle z\right)
,$$
and for $y\geq0$ and $\triangle y>0$
$$U_{\gamma}(\tau_{0}^{-}(\gamma)-)\in(y,y+\triangle y)
\Leftrightarrow
X(\tau_{\xi_{\gamma}}-)\in\left(\xi_{\gamma}(\overline{X}(\tau_{\xi_{\gamma}}))+y,\,
\xi_{\gamma}(\overline{X}(\tau_{\xi_{\gamma}}))+y+\triangle y\right).$$
\end{itemize}

The above three observations combined with Theorem 3.1 yields, for $s\geq x>0$, $y,z>0$ and $\triangle s, \triangle y,\triangle z\in(0,\infty)$
\begin{eqnarray}\label{}
\hspace{-0.3cm}&&\hspace{-0.3cm}
\mathbb{E}_{x}\left(\mathrm{e}^{-q\ell-\lambda \left(\tau_{0}^{-}(\gamma)-\ell\right)}; \,\overline{U_{\gamma}}(\tau_{0}^{-}(\gamma))\in (s,s+\triangle s),\,U_{\gamma}(\tau_{0}^{-}(\gamma)-)\in(y,y+\triangle y),
\,-U_{\gamma}(\tau_{0}^{-}(\gamma))\in(z,z+\triangle z)\right)
\nonumber\\
\hspace{-0.3cm}&=&
\hspace{-0.3cm}
\int_{\overline{s}\,\in \left(\left(\overline{\xi_{\gamma}}\right)^{-1}(s),\left(\overline{\xi_{\gamma}}\right)^{-1}(s+\triangle s)\right)}
\int_{\overline{y}\,\in\left(\xi_{\gamma}(\overline{s})+y,\,
\xi_{\gamma}(\overline{s})+y+\triangle y\right)}
\int_{\overline{z}\,\in\left(-\xi_{\gamma}(\overline{s})+z,\,
-\xi_{\gamma}(\overline{s})+z+\triangle z\right)}
\nonumber\\
\hspace{-0.3cm}&&
\times\,
\mathbb{E}_{x}\left(\mathrm{e}^{-q\ell-\lambda \left(\tau_{\xi_{\gamma}}-\ell\right)}; \,\overline{X}(\tau_{\xi_{\gamma}})\in \mathrm{d}\overline{s},\,X(\tau_{\xi_{\gamma}}-)\in\mathrm{d}\overline{y},
\,-X(\tau_{\xi_{\gamma}})\in\mathrm{d}\overline{z}\right),\nonumber
\end{eqnarray}
which combined with \eqref{pen.1} (with $\vartheta\equiv0$) yields
\begin{eqnarray}\label{}
\hspace{-0.3cm}&&\hspace{-0.3cm}
\mathbb{E}_{x}\left(\mathrm{e}^{-q\ell-\lambda \left(\tau_{0}^{-}(\gamma)-\ell\right)}; \,\overline{U_{\gamma}}(\tau_{0}^{-}(\gamma))\in \mathrm{d}s,\,U_{\gamma}(\tau_{0}^{-}(\gamma)-)\in \mathrm{d}y,
\,-U_{\gamma}(\tau_{0}^{-}(\gamma))\in\mathrm{d}z\right)
\nonumber\\
\hspace{-0.3cm}&=&
\hspace{-0.3cm}
\frac{1}{1-\gamma\left(\left(\overline{\xi_{\gamma}}\right)^{-1}(s)\right)}
\exp\left(-\int_{x}^{\left(\overline{\xi_{\gamma}}\right)^{-1}(s)}
\frac{W_{q}^{\prime}(\overline{\xi_{\gamma}}\left(w\right))}
{W_{q}(\overline{\xi_{\gamma}}\left(w\right))}\mathrm{d}w\right)
\Bigg(W_{\lambda}(0+)\,\upsilon(s+\mathrm{d}z)
\,\delta_{s}(\mathrm{d}y)
\nonumber\\
\hspace{-0.3cm}&&
\hspace{-0.3cm}
\left.
+\left(W_{\lambda}^{\prime}(s-y)
-\frac{W_{\lambda}^{\prime}(s)}{W_{\lambda}(s)}W_{\lambda}(s-y)\right)
\upsilon\left(y+\mathrm{d}z\right)
\,\mathbf{1}_{\{y<s\}}
\mathrm{d}y\right)
\mathrm{d}s,\nonumber
\end{eqnarray}
which recovers the first equation in Theorem 1.3 of Kyprianou and Zhou (2009).
Similarly,  for $s\geq x>0$ we have
\begin{eqnarray}\label{}
\hspace{-0.3cm}&&\hspace{-0.3cm}
\mathbb{E}_{x}\left(\mathrm{e}^{-q\ell-\lambda \left(\tau_{0}^{-}(\gamma)-\ell\right)}; \,\overline{U_{\gamma}}(\tau_{0}^{-}(\gamma))\in (s,s+\triangle s),\,U_{\gamma}(\tau_{0}^{-}(\gamma))=0\right)
\nonumber\\
\hspace{-0.3cm}&=&
\hspace{-0.3cm}
\int_{\overline{s}\,\in \left(\left(\overline{\xi_{\gamma}}\right)^{-1}(s),\left(\overline{\xi_{\gamma}}\right)^{-1}(s+\triangle s)\right)}
\int_{\overline{z}\,\in\{\xi_{\gamma}(\overline{s})\}}
\mathbb{E}_{x}\left(\mathrm{e}^{-q\ell-\lambda \left(\tau_{\xi_{\gamma}}-\ell\right)}; \,\overline{X}(\tau_{\xi_{\gamma}})\in \mathrm{d}\overline{s},\,X(\tau_{\xi_{\gamma}})\in\mathrm{d}\overline{z}\right),\nonumber
\end{eqnarray}
which together with \eqref{pen.2} yields
\begin{eqnarray}\label{}
\hspace{-0.3cm}&&\hspace{-0.3cm}
\mathbb{E}_{x}\left(\mathrm{e}^{-q\ell-\lambda \left(\tau_{0}^{-}(\gamma)-\ell\right)}; \,\overline{U_{\gamma}}(\tau_{0}^{-}(\gamma))\in (s,s+\triangle s),\,U_{\gamma}(\tau_{0}^{-}(\gamma))=0\right)
\nonumber\\
\hspace{-0.3cm}&=&
\hspace{-0.3cm}
\frac{1}{1-\gamma\left(\left(\overline{\xi_{\gamma}}\right)^{-1}(s)\right)}
\exp\left(-\int_{x}^{\left(\overline{\xi_{\gamma}}\right)^{-1}(s)}\frac{W_{q}^{\prime}(\overline{\xi_{\gamma}}\left(w\right))}
{W_{q}(\overline{\xi_{\gamma}}\left(w\right))}\mathrm{d}w\right)
\frac{\sigma^{2}}{2}
\left(\frac{\left(W_{\lambda}^{\prime}\left(s\right)\right)^{2}}
{W_{\lambda}\left(s\right)}
-W_{\lambda}^{\prime\prime}\left(s\right)\right)
\mathrm{d}s,\nonumber
\end{eqnarray}
which recovers the second equation in Theorem 1.3 of Kyprianou and Zhou (2009).

\medskip
\subsection{Application for L\'{e}vy risk processes with a barrier dividend strategy}

\medskip
Consider the following L\'{e}vy risk process with a barrier dividend strategy
\begin{eqnarray}
\label{3U}
R_{b}(t):=X(t)-\left(\overline{X}(t)-b\right)\vee0,
\end{eqnarray}
where $b\in(x,\infty)$ is the dividend barrier level. The risk process \eqref{3U} is well-known as the \lq\lq De Finetti's dividend model\rq\rq. For a variety of dividend risk models driven by compound Poisson processes or Brownian motions, the Gerber-Shiu function has been studied by many authors (see for example, Lin et al. (2003)), typically involves an \lq\lq infinitesimal time interval argument\rq\rq or the approach of conditioning on the time and amount of the first claim, which differs from our excursion argument. However, to the best knowledge of the authors, the Gerber-Shiu function in the context of general L\'{e}vy risk processes with a barrier dividend strategy, has not yet been studied. Here in this subsection, based on the surplus process \eqref{3U}, we attempt to express
the corresponding Gerber-Shiu function in terms of the scale functions and the L\'evy measure associated with $X$.

To fix \eqref{3U} into our drawdown setup, let
$$\xi_{b}(z):=\left(z-b\right)\vee0,\quad z\in[x,\infty),$$
which is indeed a drawdown function. One can make the following three observations.
\begin{itemize}
\item[($i^{\prime}$)]
The $\xi_{b}$-drawdown time of $X$ coincides with the ruin time of the risk process \eqref{3U}
$$\tau_{\xi_{b}}=\inf\{t\geq0; X(t)<\xi_{b}(\overline{X}(t))\}=\inf\{t\geq0; R_{b}(t)<0\}:=\tau_{0}^{-}(b).$$
\item[($ii^{\prime}$)]
The running supremum process $\{\overline{R_{b}}(t):=\sup_{0\leq s\leq t}R_{b}(s);t\geq0\}$ can be rewritten as
$$
\overline{R_{b}}(t)
=\overline{X}(t)\wedge b
=\overline{\xi_{b}}(\overline{X}(t)),$$
and hence, for small $\triangle s>0$ we have
\begin{eqnarray}
\overline{R_{b}}(\tau_{0}^{-}(b))\in[s,s+\triangle s)
\Leftrightarrow\left\{
\begin{aligned}
&\overline{X}(\tau_{\xi_{b}})\in
\left[s,s+\triangle s\right),&  x\leq s<b,\\
&\overline{X}(\tau_{\xi_{b}})\in
\left[b,\infty\right),&  x\leq s=b,\\
&\emptyset,&  \text{otherwise}.
\end{aligned}
\right.
\nonumber
\end{eqnarray}
\item[($iii^{\prime}$)] The surplus level (with dividend deducted) at and immediately before $\tau_{0}^{-}(b)$ are
$$R_{b}(\tau_{0}^{-}(b))
=X(\tau_{\xi_{b}})-\xi_{b}(\overline{X}(\tau_{\xi_{b}})),
\quad R_{b}(\tau_{0}^{-}(b)-)
=X(\tau_{\xi_{b}}-)-\xi_{b}(\overline{X}(\tau_{\xi_{b}})),$$
and hence, for $z\geq0$ and $\triangle z>0$
$$-R_{b}(\tau_{0}^{-}(b))\in[z,z+\triangle z)
\Leftrightarrow
-X(\tau_{\xi_{b}})\in\left[-\xi_{b}(\overline{X}(\tau_{\xi_{b}}))+z,\,
-\xi_{b}(\overline{X}(\tau_{\xi_{b}}))+z+\triangle z\right)
,$$
and for $y\in[0,b]$ and small $\triangle y>0$
$$R_{b}(\tau_{0}^{-}(b)-)\in[y,y+\triangle y)
\Leftrightarrow
X(\tau_{\xi_{b}}-)\in\left[\xi_{b}(\overline{X}(\tau_{\xi_{b}}))+y,\,
\xi_{b}(\overline{X}(\tau_{\xi_{b}}))+y+\triangle y\right).$$
\end{itemize}

The above observations and Theorem 3.1 yield, for $s\in[x,b)$, $y\in[x,b)$, $z>0$ and small $\triangle s, \triangle y,\triangle z\in(0,\infty)$,
\begin{eqnarray}\label{}
\hspace{-0.3cm}&&\hspace{-0.3cm}
\mathbb{E}_{x}\left(\mathrm{e}^{-q\ell-\lambda \left(\tau_{0}^{-}(b)-\ell\right)}; \,\overline{R_{b}}(\tau_{0}^{-}(b))\in [s,s+\triangle s),\,R_{b}(\tau_{0}^{-}(b)-)\in[y,y+\triangle y),
\,-R_{b}(\tau_{0}^{-}(b))\in[z,z+\triangle z)\right)
\nonumber\\
\hspace{-0.3cm}&=&
\hspace{-0.3cm}
\int_{\overline{s}\,\in \left[s,s+\triangle s\right)}
\int_{\overline{y}\,\in\left[\xi_{b}(\overline{s})+y,\,
\xi_{b}(\overline{s})+y+\triangle y\right)}
\int_{\overline{z}\,\in\left[-\xi_{b}(\overline{s})+z,\,
-\xi_{b}(\overline{s})+z+\triangle z\right)}
\nonumber\\
\hspace{-0.3cm}&&
\times\,
\mathbb{E}_{x}\left(\mathrm{e}^{-q\ell-\lambda \left(\tau_{\xi_{b}}-\ell\right)}; \,\overline{X}(\tau_{\xi_{b}})\in \mathrm{d}\overline{s},\,X(\tau_{\xi_{b}}-)\in\mathrm{d}\overline{y},
\,-X(\tau_{\xi_{b}})\in\mathrm{d}\overline{z}\right),\nonumber
\end{eqnarray}
which together with \eqref{pen.1} (with $\vartheta\equiv0$) and the fact that $\overline{\xi_{b}}(s)=s$ for $s\in[x,b]$, yield
\begin{eqnarray}\label{}
\hspace{-0.3cm}&&\hspace{-0.3cm}
\mathbb{E}_{x}\left(\mathrm{e}^{-q\ell-\lambda \left(\tau_{0}^{-}(b)-\ell\right)}; \,\overline{R_{b}}(\tau_{0}^{-}(b))\in \mathrm{d}s,\,R_{b}(\tau_{0}^{-}(b)-)\in \mathrm{d}y,
\,-R_{b}(\tau_{0}^{-}(b))\in\mathrm{d}z\right)
\nonumber\\
\hspace{-0.3cm}&=&
\hspace{-0.3cm}
\frac{W_{q}(x)}
{W_{q}(s)}
\Bigg(W_{\lambda}(0+)\,\upsilon(s+\mathrm{d}z)\,\delta_{s}(\mathrm{d}y)
\nonumber\\
\hspace{-0.3cm}&&
\hspace{-0.3cm}
\left.
+\left(W_{\lambda}^{\prime}(s-y)-\frac{W_{\lambda}^{\prime}(s)}{W_{\lambda}(s)}W_{\lambda}(s-y)\right)
\upsilon\left(y+\mathrm{d}z\right)
\,\mathbf{1}_{\{y<s\}}
\mathrm{d}y\right)
\mathrm{d}s,\quad s, y\in[x,b), z>0.\nonumber
\end{eqnarray}
For $y\in[x,b]$, $z>0$ and small $\triangle s, \triangle y,\triangle z\in(0,\infty)$
we have
\begin{eqnarray}\label{}
\hspace{-0.3cm}&&\hspace{-0.3cm}
\mathbb{E}_{x}\left(\mathrm{e}^{-q\ell-\lambda \left(\tau_{0}^{-}(b)-\ell\right)}; \,\overline{R_{b}}(\tau_{0}^{-}(b))=b,\,R_{b}(\tau_{0}^{-}(b)-)\in[y,y+\triangle y),
\,-R_{b}(\tau_{0}^{-}(b))\in[z,z+\triangle z)\right)
\nonumber\\
\hspace{-0.3cm}&=&
\hspace{-0.3cm}
\int_{\overline{s}\,\in [b,\infty)}
\int_{\overline{y}\,\in\left[\xi_{b}(\overline{s})+y,\,
\xi_{b}(\overline{s})+y+\triangle y\right)}
\int_{\overline{z}\,\in\left[-\xi_{b}(\overline{s})+z,\,
-\xi_{b}(\overline{s})+z+\triangle z\right)}
\nonumber\\
\hspace{-0.3cm}&&
\times\,
\mathbb{E}_{x}\left(\mathrm{e}^{-q\ell-\lambda \left(\tau_{\xi_{b}}-\ell\right)}; \,\overline{X}(\tau_{\xi_{b}})\in \mathrm{d}\overline{s},\,X(\tau_{\xi_{b}}-)\in\mathrm{d}\overline{y},
\,-X(\tau_{\xi_{b}})\in\mathrm{d}\overline{z}\right),\nonumber
\end{eqnarray}
which together with \eqref{pen.1} (with $\vartheta\equiv0$) yields
\begin{eqnarray}\label{}
\hspace{-0.3cm}&&\hspace{-0.3cm}
\mathbb{E}_{x}\left(\mathrm{e}^{-q\ell-\lambda \left(\tau_{0}^{-}(b)-\ell\right)}; \,\overline{R_{b}}(\tau_{0}^{-}(b))=b,\,R_{b}(\tau_{0}^{-}(b)-)\in \mathrm{d}y,
\,-R_{b}(\tau_{0}^{-}(b))\in\mathrm{d}z\right)
\nonumber\\
\hspace{-0.3cm}&=&
\hspace{-0.3cm}
\int_{s\in[b,\infty)}\exp\left(-\int_{x}^{s}\frac{W_{q}^{\prime}(\overline{\xi_{b}}\left(w\right))}
{W_{q}(\overline{\xi_{b}}\left(w\right))}\mathrm{d}w\right)
\Bigg(W_{\lambda}(0+)\,\upsilon(\overline{\xi_{b}}(s)+\mathrm{d}z)\,\delta_{\overline{\xi_{b}}(s)}(\mathrm{d}y)
\nonumber\\
\hspace{-0.3cm}&&
\hspace{-0.3cm}
\left.
+\left(W_{\lambda}^{\prime}(\overline{\xi_{b}}(s)-y)-\frac{W_{\lambda}^{\prime}(\overline{\xi_{b}}(s))}{W_{\lambda}(\overline{\xi_{b}}(s))}W_{\lambda}(\overline{\xi_{b}}(s)-y)\right)
\upsilon\left(y+\mathrm{d}z\right)
\,\mathbf{1}_{\{y<\overline{\xi_{b}}(s)\}}
\mathrm{d}y\right)
\mathrm{d}s,\quad y\in[x,b], z>0.\nonumber
\end{eqnarray}

Similarly, we have
\begin{eqnarray}\label{}
\hspace{-0.3cm}&&\hspace{-0.3cm}
\mathbb{E}_{x}\left(\mathrm{e}^{-q\ell-\lambda \left(\tau_{0}^{-}(b)-\ell\right)}; \,\overline{R_{b}}(\tau_{0}^{-}(b))\in \mathrm{d}s,\,R_{b}(\tau_{0}^{-}(b))=0\right)
\nonumber\\
\hspace{-0.3cm}&=&
\hspace{-0.3cm}
\frac{\sigma^{2}}{2}\left(\frac{W_{q}(x)}
{W_{q}(s)}
\left(\frac{\left(W_{\lambda}^{\prime}\left(s\right)\right)^{2}}
{W_{\lambda}\left(s\right)}
-W_{\lambda}^{\prime\prime}\left(s\right)\right)\mathbf{1}_{[x,b)}(s)
\,\mathrm{d}s\right.
\nonumber\\
\hspace{-0.3cm}&&
\hspace{-0.3cm}
\left.+\,
\delta_{b}(\mathrm{d}s)\int_{b}^{\infty}\exp\left(-\int_{x}^{z}\frac{W_{q}^{\prime}(\overline{\xi_{b}}\left(w\right))}
{W_{q}(\overline{\xi_{b}}\left(w\right))}\mathrm{d}w\right)
\left(\frac{\left(W_{\lambda}^{\prime}\left(\overline{\xi_{b}}\left(z\right)\right)\right)^{2}}
{W_{\lambda}\left(\overline{\xi_{b}}\left(z\right)\right)}
-W_{\lambda}^{\prime\prime}\left(\overline{\xi_{b}}\left(z\right)\right)\right)
\mathrm{d}z\right).\nonumber
\end{eqnarray}

\section{Numerical examples}

The results in Section 3 are illustrated with several examples in this section. One quantity of interest is the probability of drawdown, which includes the probability of ruin as a special case. The other item is the joint density of the drawdown time and the first time when the running maximum prior to the drawdown time is hit.

\subsection{A Comparison of Drawdown Probabilities and Ruin Probabilities}

 The numerical results in this subsection is based on equation (\ref{GS.f.}) where we let the penalty function be $\omega(x,y)\equiv 1$ and the discount factor $q\equiv 0$. Then the expected penalty function at general drawdown is specialized to $\phi(x)=
\mathbb{E}_{x}\left(\mathbf{1}_{\{\tau_{\xi}<\infty\}}\right)=\mathbb{P}_{x}(\tau_{\xi}<\infty)$, which refers to the probability of general drawdown. For simplicity, we consider the general drawdown function in a linear form $\xi(x)=ax-b$. According to the definition of general drawdown function in Section 2, we require $a\in(-\infty,1)$ and $b>0$. Then $X(t)<\xi(\overline{X}(t))=a \overline{X}(t)-b$ is equivalent to $a \overline{X}(t)-X(t)>b$, which refers to the surplus process drops $b$ units below $100a$ percent of its maximum to date. In this section, we compare the evolution of the probability of drawdown according to the following four sets of parameters.

\begin{table}[!hp] \vspace{-8pt}\caption{The parameters for the linear drawdown function.}
\renewcommand{\arraystretch}{1.1}\centering \setlength{\tabcolsep}{4pt} 
\begin{tabular}{|c|c|c|}
  \hline
   & $a$ & $b$ \\ \hline
  (I) & 0 & 0 \\  \hline
  (II) & 0.3 & 0.5 \\  \hline
  (III) & 0.5 & 0.5 \\ \hline
  (IV) & 0.6& 0.5 \\
  \hline
\end{tabular}
\end{table}

It is obvious from the table that (I) refers to the ruin case. Starting from the value $X(0)=x$, say $x=1$, suppose that at some time point $t>0$ we have $X(t)<0$, then $0.6\overline{X}(t)-X(t)\geq 0.6-X(t)>0.6>0.5$. That is, the drawdown time for case (IV) must occur before the ruin time, which is denoted by $\tau_{IV}<\tau_{I}$. Similarly, we have $\tau_{IV}<\tau_{III}<\tau_{II}$ and $\tau_{IV}<\tau_{III}<\tau_{I}$. Accordingly, the sooner the general drawdown time occurs in theory, the higher occurring probability corresponding to it. We point out that the relation between case (I) and case (II) is not clear, which depends on the underlying risk process and the parameter setting. We can see the difference in the following examples.

\begin{exa}
{Cram\'e}r-Lundberg model with exponential jumps.
{\rm Suppose that the {L\'e}vy process is given by the {Cram\'e}r-Lundberg model with exponential jumps. To be specific, when $X(t)$ is reduced to a compound Poisson process with Poisson arrival rate $\lambda_0>0$, premium rate $c$, and claim sizes following an exponential distribution with mean $1/\mu>0$,
$$X(t)=x+ct-\sum_{i=1}^{N(t)}Y_{i},\quad t\geq0,$$
then the expected discounted penalty function at the general drawdown time, that is the probability of general drawdown is given by (\ref{GS.f.}) with $\upsilon(\mathrm{d}z)=\lambda_0 F(\mathrm{d}z)=\lambda_0 \mu e^{-\mu z}\mathrm{d}z$ and $\sigma=0$. The explicit expression for the scale function is available as,
\begin{eqnarray}
W_{q}(x)=\frac{A_1(q)}{c}\mathrm{e}^{\theta_1(q)x}
-\frac{A_2(q)}{c}\mathrm{e}^{\theta_2(q)x},\quad x\geq0,\label{Ex2_W}
\nonumber
\end{eqnarray}
with $A_1(q)=\frac{\mu+\theta_1(q)}{\theta_1(q)-\theta_2(q)}$ and $A_2(q)=\frac{\mu+\theta_2(q)}{\theta_1(q)-\theta_2(q)}$, where $\theta_1(q)
=\frac{\lambda_0+q-c\mu+\sqrt{(c\mu-\lambda_0-q)^2+4cq\mu}}{2c}$ and $\theta_2(q)
=\frac{\lambda_0+q-c\mu-\sqrt{(c\mu-\lambda_0-q)^2+4cq\mu}}{2c}$. Due to $\sigma=0$, we only need to compute the first two lines of integrals in (\ref{GS.f.}), which represent the risk measurement brought by exponentially distributed claims. We are interested in the impact of the initial capital $x$ and the premium rate $c$ on the probability of drawdown.

We list the parameter values as follows: $x=1$, $c=1.1$ and $\lambda_0=\mu=2$. In the classical risk theory, the higher $x$ or $c$, the lower probability of ruin. This can be verified by Figure 1. We can also observe a similar trend for the probability of drawdown, which is decreasing along with $x$ or $c$. Figure 1 also verify our prior analysis on the relations between the results of different drawdown functions. That is, case (IV) has the highest value, followed by case (III), then case (II) and (I). Unlike the unclear theoretical comparison between case (I) and case (II), we observe that case (II) stays on top of case (I) in Figure 1. This may due to the exponentially distributed claim size nature of the underlying risk process, which results in a similar relation for the jump-diffusion model in Figure 3.

\begin{figure}[!hp]
\centerline{\includegraphics[width=15cm,
height=10cm]{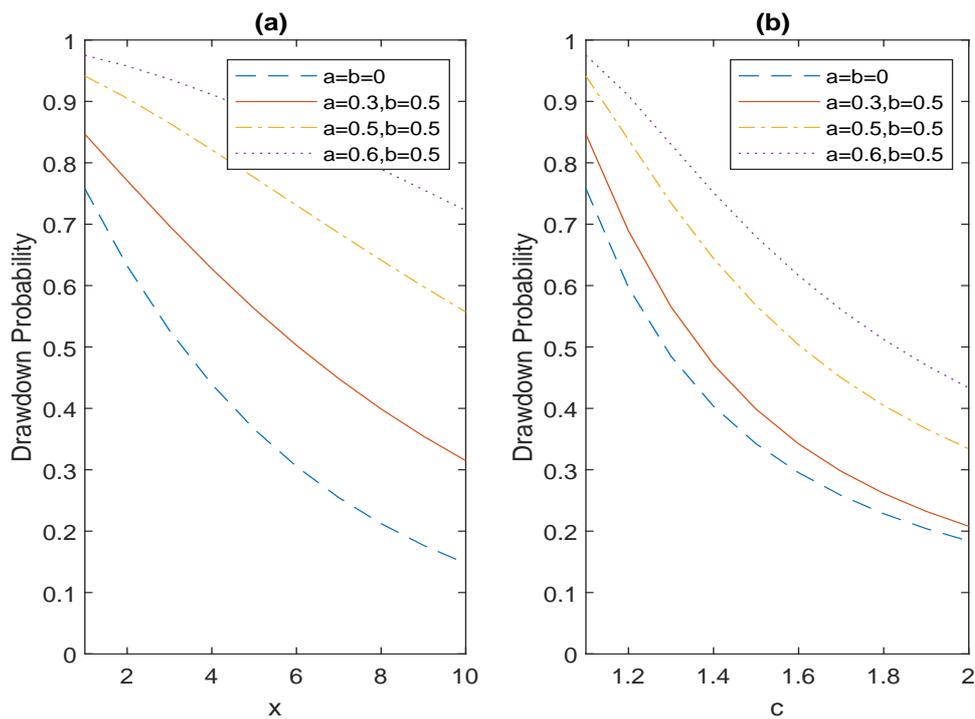}}
\caption{(a) The drawdown probability as a function of $x\in[1,10]$ for $c=1.1$  and (b) as a function of $c\in [1.1, 2.1]$ for $x=1 $ in compound Poisson process.}
\label{fig1}%
\end{figure}

Another observation from Figure 1 is, the probability of drawdown is more sensitive to the change of premium rate $c$ than the initial capital $x$. Similar slops are produced from part $(a)$ as $x$ increases from unit 1 to unit 10, and from part (b) as $c$ increases from 1.1 to 2.1. This gives the risk manager a hint that, the adjustment on the premium rate has a more immediate effect on the risk level than an injection of capital.
}
\end{exa}

\begin{exa}\label{Bexa} Brownian motion with drift.
{\rm Brownian motion (with or without drift) is the only continuous L{\'e}vy process. When $X$ is reduced to a Brownian motion with drift
$$X_{t}=x+ \mu t+\sigma B_{t},\quad t\geq0,\,\, \mu\neq0,\,\,\sigma>0.$$
Then the explicit expression for the scale function is written as
\begin{eqnarray}\label{15}
W_{q}(x)=a_0(e^{\lambda_{1}x}
-e^{\lambda_{2}x}),\quad x\geq0,
\end{eqnarray}
where $a_0=(2q\sigma^{2}+\mu^{2})^{-\frac{1}{2}}$, $\lambda_{1}=\frac{(2q\sigma^{2}+\mu^{2})^{\frac{1}{2}}-\mu}{\sigma^{2}}$ and $\lambda_{2}=\frac{-(2q\sigma^{2}+\mu^{2})^{\frac{1}{2}}-\mu}{\sigma^{2}}$. Under this continuous risk process, the L{\'e}vy measure $\upsilon(\mathrm{d}z)=0$. Then the expected discounted penalty function at general drawdown time is given by the third line only in equation (\ref{GS.f.}).

\begin{figure}[!hp]
\centerline{\includegraphics[width=15cm,
height=10cm]{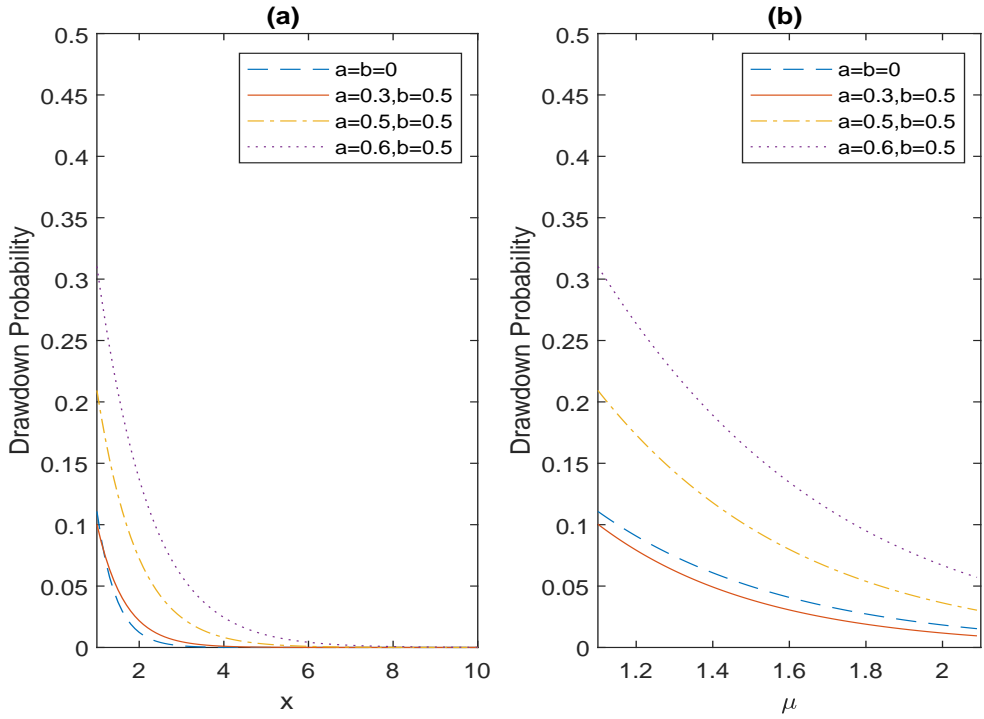}}
\caption{(a) The drawdown probability as a function of $x\in[1,10]$ for $\mu=1.1$  and (b) as a function of $\mu\in [1.1, 2.1]$ for $x=1 $ in Brownian motion process.}
\label{fig2}%
\end{figure}

Taking $\sigma=1$ which is a relative low volatility level comparing to the drift parameter $\mu$ increasing from 1.1 to 2.1 as plotted in Figure 2. Accordingly, we expect a low probability of ruin as well as the probability of general drawdown. This can be seen in Figure 2 that the levels are much lower than those in Figure 1. Comparing to Figure 1, the overall trends in Figure 2 are similar. The only difference is the blur relation between case (I) and case (II), which is also blur in theory and may comes from the small fluctuations described by the Brownian motion.}
\end{exa}

\begin{exa}Jump-diffusion process.
{\rm When $X$ is reduced to a jump-diffusion process,
\[{X(t)} = x + ct + \sigma {W_t} - \sum\limits_{i = 1}^{{N(t)}} {{Y_i}} {\rm{ }},\quad t\geq 0,\]
where $\sigma {\rm{ > 0}}$, $ \left\{ {{N(t)},t \ge 0} \right\}$  is a Poisson process with arrival rate $\lambda_0$, and $Y_{i}$'s  are a sequence of i.i.d. random variables distributed with Erlang $ \left( {2,\alpha } \right).$ The scale function associated with $X$ can be derived as (cf., Loeffen (2008))
\begin{align}\label{expressionsofWqin2}
{W_{q}}\left( x \right) = \sum\limits_{j = 1}^4 {{D_j}\left( q \right){\mathrm{e}^{{\theta _j}\left( q \right)x}}} ,{\kern 1pt} {\kern 1pt} {\kern 1pt} {\kern 1pt} {\kern 1pt} {\kern 1pt} {\kern 1pt} {\kern 1pt} {\kern 1pt} {\kern 1pt} {\kern 1pt} x \ge 0,
 \end{align}
where
\begin{align}\label{expressionsofdj}
{D_j}\left( q \right) = \frac{{{{\left( {\alpha  + {\theta _j}\left( q \right)} \right)}^2}}}{{\frac{\sigma ^2}{2}\prod\limits_{i = 1,i \ne j}^4 {\left( {{\theta _j}\left( q \right) - {\theta _i}\left( q \right)} \right)} }},
\nonumber
 \end{align}
and ${\theta _j}\left( q \right)\left( {j = 1, \cdots 4} \right)$ are the (distinct) zeros of the polynomial
\begin{align}
\nonumber &\quad\Big(c\theta-\lambda_0+\frac{\lambda_0\alpha^2}{(\alpha+\theta)^2}+\frac{1}{2}\sigma^2\theta^2-q\Big)(\alpha+\theta)^2\\
\nonumber &=\frac{1}{2}\sigma^2\theta^4+(\alpha\sigma^2+c)\theta^3+\big(\frac{1}{2}\sigma^2\alpha^2-\lambda_0-q+2c\alpha\big)\theta^2+\big[c\alpha^2-2(\lambda_0+q)\alpha\big]\theta-q\alpha^2.
\end{align}
Then the expected discounted penalty function at the general drawdown is given by equation (\ref{GS.f.}) with $\upsilon(\mathrm{d}z)=\lambda_0 \alpha^{2} z \mathrm{e}^{-\alpha z} \mathrm{d}z$ and $W_{q}$ given by (\ref{expressionsofWqin2}).

\begin{figure}[!hp]
\centerline{\includegraphics[width=15cm,
height=10cm]{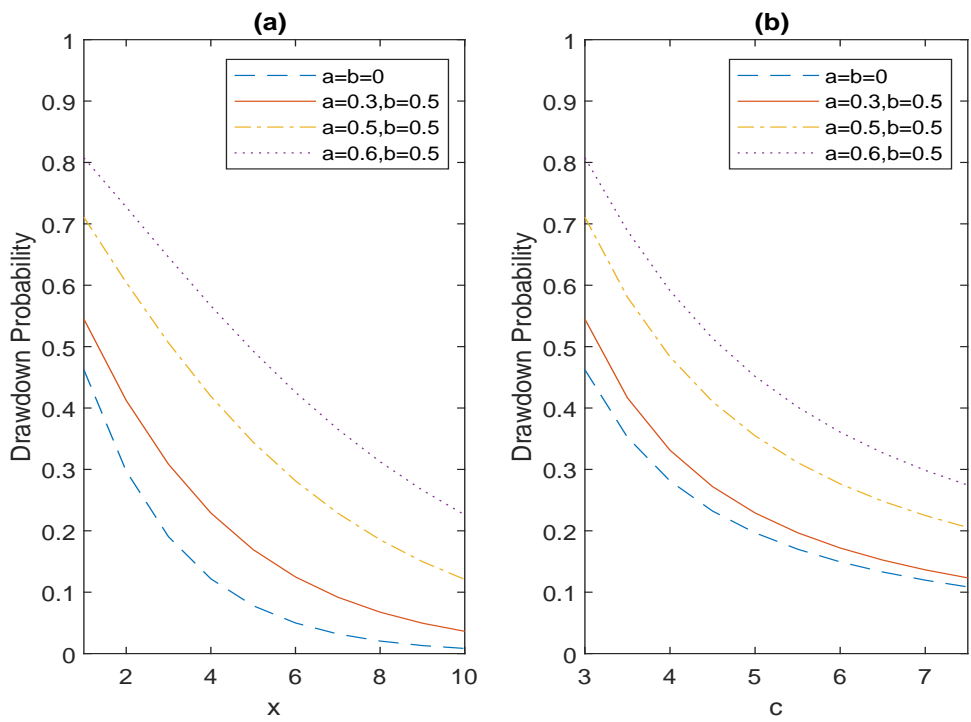}}
\caption{(a) The drawdown probability as a function of $x\in[1,10]$ for $c=3$  and (b) as a function of $c\in [3, 7.5]$ for $x=1 $ in jump diffusion process.}
\label{fig3}%
\end{figure}

The parameters in this example are: $\sigma=0.5$ and $\lambda_0=\alpha=2$. Comparing to the aforementioned two examples, we have neither $\nu(dz)=0$ nor $\sigma=0$ for the jump-diffusion process, which results in the involvement of all the three lines of equation (\ref{GS.f.}) in our computation. Then it is natural to expect a much higher probability of drawdown than the previous two examples. Intuitively speaking, the drawdown of the jump-diffusion process is composed of two parts: the small fluctuations described by the Brownian motion $W_t$ and the large jumps described by the compound Poisson process $\sum_{i = 1}^{{N(t)}} {{Y_i}}$. This has been verified in our computation that we have to choose a much higher premium rate $c$ than the previous examples, otherwise the ruin probability would be 1. In Figure 3, we let the initial capital increase from 1 to 10 in Part $(a)$, and premium rate increase from 3 to 7.5 in part (b), then we produce a similar trend as in Figure 1.}
\end{exa}

\subsection{Joint Density of $\tau_{\xi}$ and $\ell$. }

The first time when the running maximum prior to $\tau_{\xi}$ is hit is denoted by $\ell$ in this paper. This subsection intends to study the joint distribution of the drawdown time $\tau_{\xi}$ and $\ell$, as well as the effects of the model parameters. By Theorem \ref{3.1}, the joint Laplace transform of $\ell$ and $\tau_{\xi}$ is
\begin{eqnarray}\label{BLap}
\hspace{-0.3cm}&&\hspace{-0.3cm}
\mathbb{E}_{x}\left(\mathrm{e}^{-q\ell-\lambda \left(\tau_{\xi}-\ell\right)}\,\mathbf{1}_{\{\tau_{\xi}<\infty\}}\right)
\nonumber\\
\hspace{-0.3cm}&=&
\hspace{-0.3cm}
\int_{s\in(x,\infty)}\int_{y\in[\xi(s),s]}\int_{z\in(-\xi(s),\infty)}\exp\left(-\int_{x}^{s}\frac{W_{q}^{\prime}(\overline{\xi}\left(w\right))}
{W_{q}(\overline{\xi}\left(w\right))}\mathrm{d}w\right)
\Bigg(W_{\lambda}(0+)\,\upsilon(s+\mathrm{d}z)\,\delta_{s}(\mathrm{d}y)
\nonumber\\
\hspace{-0.3cm}&&
\hspace{-0.3cm}
\left.
+\left(W_{\lambda}^{\prime}(s-y)-\frac{W_{\lambda}^{\prime}(\overline{\xi}(s))}
{W_{\lambda}(\overline{\xi}(s))}W_{\lambda}(s-y)\right)
\upsilon\left(y+\mathrm{d}z\right)
\,\mathbf{1}_{\{y<s\}}
\mathrm{d}y\right)
\mathrm{d}s
\nonumber\\
\hspace{-0.3cm}&&\hspace{-0.3cm}
+
\frac{\sigma^{2}}{2}\int_{s\in(x,\infty)}\exp\left(-\int_{x}^{s}\frac{W_{q}^{\prime}(\overline{\xi}\left(w\right))}
{W_{q}(\overline{\xi}\left(w\right))}\mathrm{d}w\right)
\left(\frac{\left(W_{\lambda}^{\prime}\left(\overline{\xi}\left(s\right)\right)\right)^{2}}
{W_{\lambda}\left(\overline{\xi}\left(s\right)\right)}
-W_{\lambda}^{\prime\prime}\left(\overline{\xi}\left(s\right)\right)\right)
\mathrm{d}s.
\end{eqnarray}

For computational simplicity, our numerical results are based on {\it Brownian motion with drift} Example \ref{Bexa} . The main parameters are $x=1$, $\mu=0.3$ and $\sigma=1$. In terms of $\tau_{\xi}$ we still use the linear drawdown function defined in the previous subsection. And, we only consider two cases for the parameters: the ruin case with $a=b=0$ and the drawdown case with $a=0.6$ and $b=0.5$. In the following, we take equation (\ref{BLap}) as the joint Laplace transform of $\tau_{\xi}$ and $\ell$, and then take the inverse transformation to derive the joint density distribution. Our algorithm is based on the method of Fourier series expansion proposed by Moorthy (1995), which is one of the most worthwhile methods in the numerical inversion of Laplace transforms. There are also many other basic methods available in the literature, say for example, the Laguerre function expansion and Combination of Gaver functions. These basic methods breed over 100 algorithms on the subject. As we have explained in Section 1, it is hard to find a universal algorithm that works for all the cases and we are not intending to develop a perfect algorithm in this paper. Therefore, we just follow the Fourier series expansion method to write our codes in Matlab, and accept the instability behavior at the boundaries of the produced results.

In Figure 4, $f(t_1, t_2)$ refers to the joint density function of random variables $\tau_{\xi}$ (corresponds to $t_1$) and $\ell$ (corresponds to $t_2$). We make the following observations. Firstly, the overall trend of $f(t_1, t_2)$ goes to zero as $t_1$ or $t_2$ gets bigger and bigger indicating that the drawdown time or the running maximum hitting time  occurring at a later time has a smaller probability, which is consistent with the path behavior of a \emph{Brownian motion with positive drift}. In fact, due to $\mu>0$ we have $\lim\limits_{t\rightarrow\infty}X(t)=\infty$, which means that, if drawdown occurs, it occurs at a finite time, and less and less likely to occur at a later time until infinity. Secondly, since $\ell$ refers to a time point that is prior to $\tau_{\xi}$, then the value of the probability density $f(t_1,t_2)=0$ for $t_1<t_2$. This can be seen clearly in Figure 4 that all the positively valued $f(t_1,t_2)$ are distributed on the side of $t_1$ axis. Thirdly, the drawdown time is expected to occur before the ruin time, then the value of $f(t_1,t_2)$ is more concentrated at smaller values of $t_1$ and $t_2$ in the drawdown case than in the ruin case. We can see clearly that in Figure 4 that the drawdown case in (b) builds up a higher value than the ruin case in (a).

\begin{figure}[!hp]
\centerline{\includegraphics[width=15cm,
height=10cm]{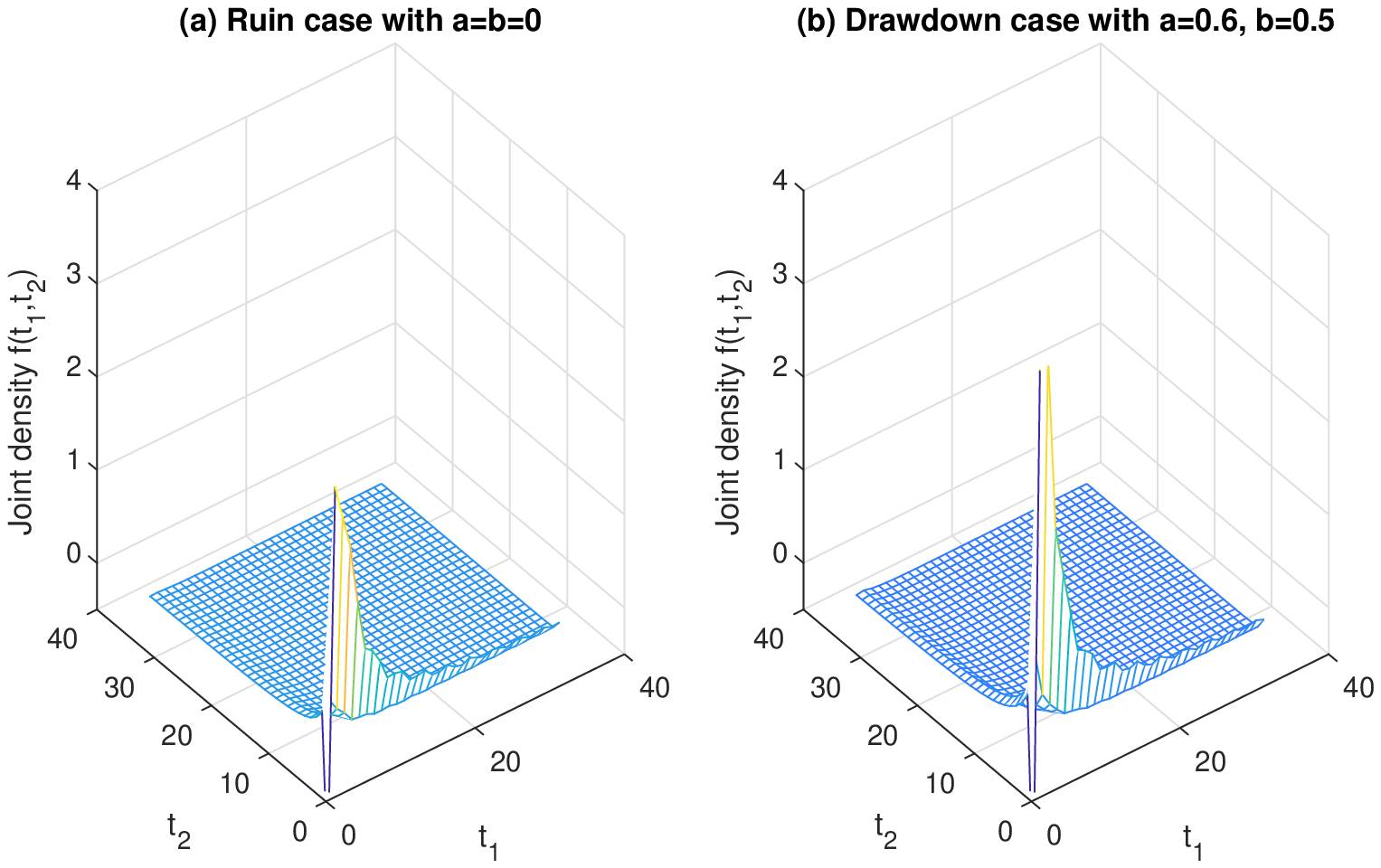}}
\caption{ The joint density $f(t_1,t_2)$ for $\tau_{\xi}$ and $\ell$.}
\label{fig4}%
\end{figure}

Next we look at the effects of parameters $x$, $\mu$ and $\sigma$ on the value of $f(t_1,t_2)$. The basic parameters are $x=1$, $\mu=0.3$, $\sigma=1$, $a=0.6$ and $b=0.5$, which lead to (a) of Figure 5. When we change the initial capital $x$ from 1 to 2 in (b), the drawdown time is expected to occur later intuitively, and correspondingly we observe a lower but fatter joint density function in (b). Similarly, a bigger value of $\mu$ helps to build up the surplus value which in turn leads to a later drawdown time. We also observe a lower but fatter distribution in (c) comparing to (a). The effect of $\sigma$ goes to the other direction, the resulted joint density in (d) is higher and sharper comparing to (a), which explains the higher uncertainty brought by a larger value of $\sigma$.

\begin{figure}[!hp]
\centerline{\includegraphics[width=15cm,
height=10cm]{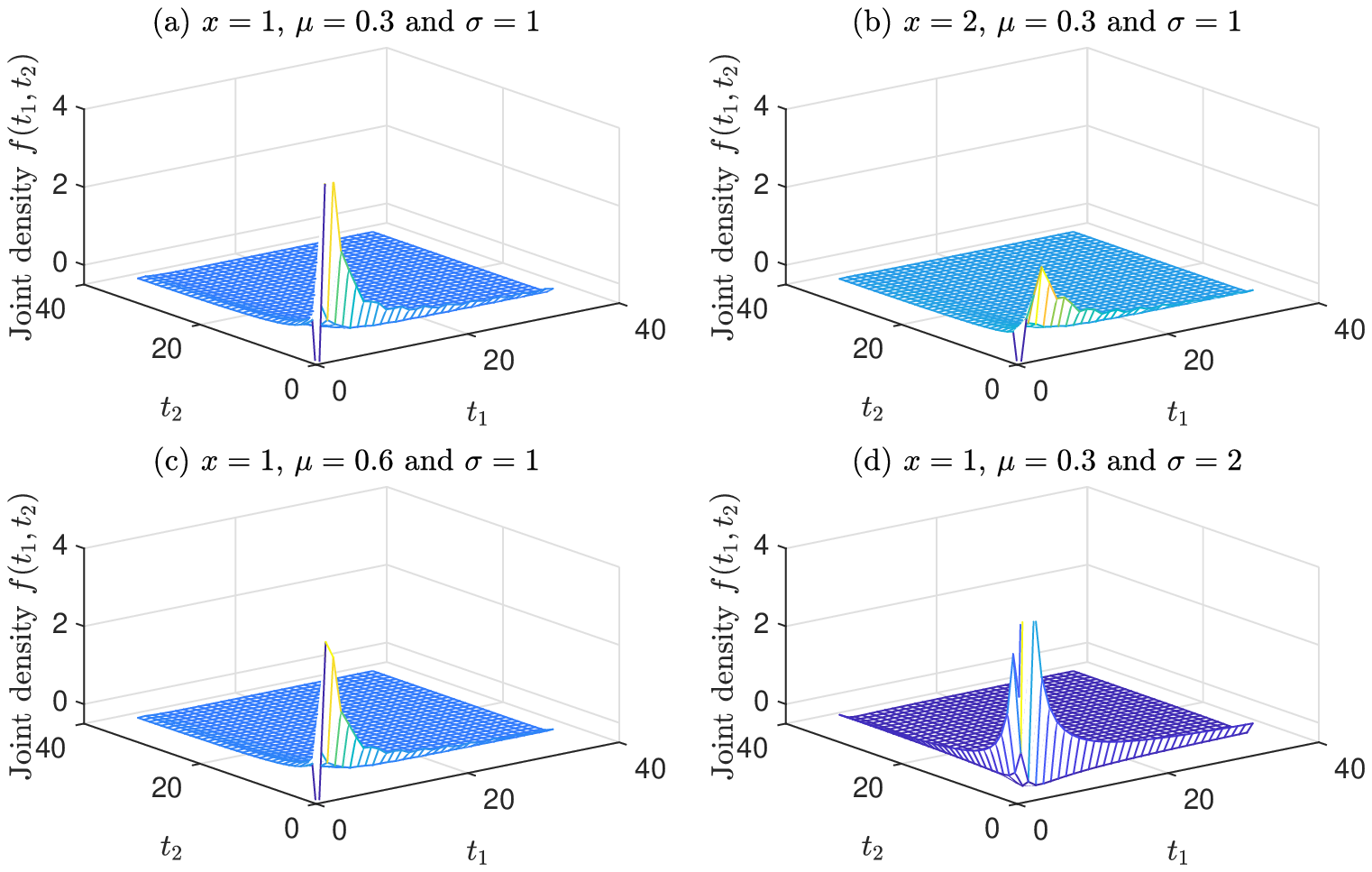}}
\caption{ The joint density $f(t_1,t_2)$ for $\tau_{\xi}$ and $\ell$.}
\label{fig5}%
\end{figure}

\section{Conclusion}

In this paper, the generalized Gerber-Shiu function at general drawdown time is considered for a spectrally negative L\'{e}vy process. It is shown that in the present model, the extended Gerber-Shiu function can be expressed in terms of the $q$-scale functions and the L\'{e}vy measure associated with the L\'{e}vy process. This expression makes it possible to study the joint distribution of the time of drawdown, the running maximum at drawdown, the last minimum before drawdown, the surplus before drawdown and the surplus at drawdown, which broaden the family of risk indicators and measurements. The motivation of such an extension from the time of ruin to the time of drawdown is two folds. First, thanks to the development of the excursion approach in solving boundary crossing problems related with L\'{e}vy processes, such that the derivation of the extended Gerber-Shiu function is possible. Second, the time of drawdown has a clearer description of the company's financial position than the time of ruin. Then the insurer can take actions more promptly and effectively, such as adjusting the premium rate or injecting more capital to prevent even worse situations.




\vspace{1cm}
\section*{Acknowledgements}
The authors are very grateful to 
Professor Xiaowen Zhou at Concordia University for his helpful comments on 
this paper.

\bibliographystyle{elsarticle-num}
\bibliography{reference}

\end{document}